\documentclass{article}

 \usepackage[a4paper,bindingoffset=0.2in,             left=0.9in,right=0.9in,top=1in,bottom=1in,
             footskip=.25in]{geometry}

\usepackage[font=small,skip=0pt]{caption}

\usepackage{pgfplots} 
\usetikzlibrary{arrows,matrix,positioning}
\usepackage[normalem]{ulem}
\usepgfplotslibrary{fillbetween}
\usepackage{hyperref}
\usepackage{amssymb}
\usepackage{amsfonts} 
\usepackage{amsmath} 
\usepackage{url}
\usepackage{cite}
\usepackage{amsthm}
\theoremstyle{plain}
\newtheorem{theorem}{Theorem}
\newtheorem{corollary}{Corollary}
\newtheorem{lemma}{Lemma}
\newtheorem{proposition}{Proposition}
\theoremstyle{definition}

\theoremstyle{remark}
\newtheorem{remark}{Remark}
\usepackage{algorithm,algcompatible,lipsum}
\usepackage{color}
\usepackage{tcolorbox}
\usepackage{todonotes}
\newcommand{\bs}{\boldsymbol}
\newcommand{\mbf}{\mathbf}
\newcommand{\mbb}{\mathbb}
\newcommand{\mcl}{\mathcal}
\newcommand{\mrm}{\mathrm} 
\usepackage{xstring}
\def\alphabet{abcdefghijklmnopqrstuvwxyzABCDEFGHIJKLMNOPQRST123456789}
\renewcommand{\vec}[1]{
\IfSubStr{\alphabet}{#1}{
\ensuremath{\mathbf{\MakeLowercase{#1}}}
}{
\ensuremath{\boldsymbol{\MakeLowercase{#1}}}
}
}
\newcommand{\mat}[1]{
\IfSubStr{\alphabet}{#1}{
\ensuremath{\mathbf{\MakeUppercase{#1}}}
}{
\ensuremath{\boldsymbol{\MakeUppercase{#1}}}
}
}

\def\R{\mathbb R}

\def\N{\mathbb N}

\newcommand*{\inner}[2]{\left\langle#1,#2\right\rangle}
\newcommand*{\norm}[1]{\left\|#1\right\|}
\newcommand*{\paran}[1]{\left(#1\right)}

\newcommand*{\card}[1]{\left|#1\right|}
\def\defeq{:=}
\def\d{\mathrm{d}} 

\def\S{\mathbb{S}} 
\def\SO{\mathrm{SO}} 
\newcommand*{\Y}[2]{\mathrm{Y}_{#1}^{#2}} 
\newcommand*{\D}[3]{\mathrm{D}_{#1}^{#2,#3}} 
\newcommand*{\Wd}[3]{\mathrm{d}_{#1}^{#2,#3}} 
\def\i{\mathrm{i}}

\usepackage{xcolor}

\usepackage{acronym}
\acrodef{CS}{Compressed Sensing}
\acrodef{BP}{Basis Pursuit}
\acrodef{OMP}{Orthogonal Matching Pursuit}
\acrodef{AMP}{Approximate Message Passing}
\acrodef{RIP}{Restricted Isometry Property}
\acrodef{BOS}{Bounded Orthonormal System}
\acrodef{IGRF}{International Geomagnetic Reference Field}
\begin{document}

\title{
Tight bounds on  the mutual coherence of sensing matrices for Wigner D-functions on regular grids 
}

\author{Arya~Bangun,
        Arash~Behboodi,
        and~Rudolf~Mathar,
\thanks{Institute for Theoretical Information Technology (TI), RWTH Aachen University.}}

\maketitle
\begin{abstract}

Many practical sampling patterns for function approximation on the rotation group utilizes regular samples on the parameter axes.  In this paper, we relate the mutual coherence analysis for sensing matrices that correspond to a class of regular patterns to angular momentum analysis in quantum mechanics and provide simple lower bounds for it. The products of Wigner d-functions, which appear in coherence analysis, arise in angular momentum analysis in quantum mechanics.
We first represent the product as a linear combination of a single Wigner d-function and angular momentum coefficients, otherwise known as the Wigner 3j symbols. 
Using combinatorial identities, we show that under certain conditions on the bandwidth and number of samples, 
the inner product of the columns of the sensing matrix at zero orders, which is equal to the inner product of two Legendre polynomials, dominates the mutual coherence term and fixes a lower bound for it. In other words, for a class of regular sampling patterns, we provide a lower bound for the inner product of the columns of the sensing matrix that can be analytically computed. We verify  numerically  our theoretical results and show that the lower bound for the mutual coherence is larger than Welch bound. Besides, we provide algorithms that can achieve the lower bound for spherical harmonics. 
\end{abstract}


\section{Introduction}
In many applications, the goal is to recover a function defined on a group, say on the sphere $\S^2$ and the rotation group $\SO(3)$, from only a few samples \cite{bangun_sensing_2020,bangun_coherence_2018,culotta-lopez_compressed_2018, burq_weighted_2012, rauhut_sparse_2011}. This problem can be seen as a linear inverse problem with structured sensing matrices that contain samples of spherical harmonics and Wigner D-functions. 

In the area of compressed sensing and sparse signal recovery, there are recovery guarantee results for random sampling patterns \cite{rauhut_sparse_2011, burq_weighted_2012} based on proving \ac{RIP} of the corresponding sensing matrix. Regular deterministic sampling patterns are, however, more prevalent in practice due to their easier deployment \cite{bangun_coherence_2018, bangun_sensing_2020}. \ac{RIP} based results cannot be used for analyzing deterministic sensing matrices, because it has been shown in \cite{tillmann2014computational, bandeira2013certifying} that verifying \ac{RIP} is computationally hard. For deterministic sampling patterns, the mutual coherence is widely used as performance indicator, which measures the correlation between different columns in the sensing matrix. For the case of sparse recovery on $\S^2$ and $\SO(3)$, the product of two Wigner D-functions and the product of two spherical harmonics appear in the coherence analysis. 

The product of orthogonal polynomials, i.e., Legendre and Jacobi polynomials, is sought-after in mathematics, and it is related to hypergeometric functions. Several works to obtain a closed-form and simplified version of these products have been presented, for instance in   \cite{dougall_product_1953,adams_iii_1878,gasper_linearization_1970,george_gasper_linearization_1970 ,askey_linearization_1971}. In those articles, the authors attempt to derive a compact formulation of the product and represent it as a linear combination of a single orthogonal polynomial with some coefficients. This representation reveals interesting properties that can also be applied in quantum mechanics \cite{edmonds_angular_1974, rose_elementary_1995}, geophysics \cite{simons_spatiospectral_2006}, machine learning \cite{kondor_clebsch_2018}, and low-coherence sensing matrices \cite{bangun_sensing_2020,bangun_coherence_2018,culotta-lopez_compressed_2018}.  In quantum mechanics, these coefficients are used to calculate
the addition of angular momenta, i.e., the interaction between
two charged particles, and such coefficients are called the
Clebsch-Gordan coefficients or Wigner 3j symbols. Specifically, these coefficients appear when we want to determine
the product of Wigner D-functions
and spherical harmonics.
Since Jacobi and Legendre polynomials are specifically used
to express the Wigner D-functions and spherical harmonics, the product of those polynomials in terms of Wigner 3j symbols is also commonly used in angular momentum literature \cite{edmonds_angular_1974, rose_elementary_1995}. The product of spherical harmonics also emerges in the
spatiospectral concentration or Slepian's concentration problem on the sphere, as discussed in \cite{simons_spatiospectral_2006}. In this case, the goal is to maximize the concentration of spectrum frequencies, i.e., spherical harmonic coefficients, given a certain area on the spherical surface. The problem is similar to finding relevant eigenvalues from a matrix that consists of the product of spherical harmonics. As a consequence, Wigner 3j symbols appear as a tool to analyze the problem. In the area of machine learning, Wigner D-functions are used for analyzing group transformations of the input to neural networks and to implement equivariant architectures for rotations. The authors in \cite{kondor_clebsch_2018} develop \textit{Clebsch-Gordan nets} to generalize and to improve the performance of spherical convolutional neural networks for recognizing the rotation of spherical images, 3D shape, as well as predicting energy of the atom. The contribution of this article includes the implementation of Wigner D-functions to perform the transformation of a signal in the Fourier domain and to tailor a representation of the product in terms of Clebsch-Gordan coefficients that will primarily support theoretical analysis of the networks. 

In this work, we are interested in coherence analysis and sparse recovery tasks.  
It has been shown in  \cite{bangun_sensing_2020} that a wide class of modular symmetric regular sampling patterns, like equiangular sampling, yield highly coherent sensing matrices and thereby perform poorly for signal recovery tasks. Besides, the coherence was shown to be affected by the choice of elevation sampling pattern independent of azimuth and polarization sampling patterns. It was numerically shown that for regular sampling points on the elevation for Wigner D-functions and spherical harmonics, the mutual coherence is lower bounded by the inner product of columns with zero orders and two largest degrees\footnote{
To simplify the presentation, the order and degree of a column refers to the order and degree of the respective Wigner D-functions or spherical harmonics.
}, which are then equal to Legendre polynomials. This bound is not contrived because one can show that this bound is achievable by optimizing azimuth angle $\phi \in [0, 2\pi)$. Consequently, the resulting deterministic sampling points can be implemented into a real-system to carry out measurements on the spherical surface, as discussed in \cite{culotta-lopez_compressed_2018,culotta-lopez_practical_2019}. In this article, we confirm mathematically the numerical findings of \cite{bangun_sensing_2020}. Our proof relies on using results for angular momentum analysis in quantum mechanics and properties of Wigner 3j symbols. To the best of our knowledge, this work is the first to provide the coherence analysis of a sensing matrix using the tools from angular momentum in quantum mechanics.

\subsection{Related Works}
The construction of sensing matrices from a set of orthogonal polynomials is widely studied in the area of compressed sensing. For instance, in \cite{rauhut_sparse_2012} the authors show that the sensing matrix construction from random samples of Legendre polynomials with respect to Chebyshev measure fulfils the \ac{RIP} condition and thus performs a robust and stable recovery guarantee to reconstruct sparse functions in terms of Legendre polynomials by using $\ell_1$-minimization algorithm. The extensions for random samples spherical harmonics and Wigner D-functions are discussed in \cite{rauhut_sparse_2011,burq_weighted_2012,rauhut_interpolation_2016,bangun_sensing_2020}. The key idea in those articles bears a strong resemblance to the technique discussed in Legendre polynomials' case, which is based on carefully choosing random samples with respect to different probability measures and preconditioning techniques to keep the polynomials uniformly bounded. Despite the recovery guarantees with regard to the minimum number of samples, it has been discussed in \cite{hofmann_minimum_2019} that these theoretical results is seemingly too pessimistic. Practically, one can consider a smaller number of samples and still achieve a very good reconstruction by using $\ell_1$-minimization algorithm. Therefore, a gap between theoretical and practical settings exists. Another concern in antenna measurement system is designing a smooth trajectory for robotic arms to acquire electromagnetic fields, which causes a practical obstacle in using random samples, as mentioned in \cite{culotta-lopez_compressed_2018,culotta-lopez_practical_2019}. 

One of the most prevalent application of orthogonal polynomials is in the area of interpolation, where those polynomials are used to approximate a function within a certain interval. Recently, the  $\ell_1$-minimization-based technique is tailored to interpolate a function which has sparse representation in terms of Legendre and spherical harmonic coefficients, as discussed in \cite{rauhut_interpolation_2016}. In this case, random samples of Legendre and spherical harmonics are used to construct sensing matrices. Similar to the results in compressed sensing, the \ac{RIP} plays a pivotal role in showing that a particular number of samples is required to achieve certain error approximation. 

Another construction of sensing matrices from some orthogonal polynomials related to the sparse polynomial chaos expansion is investigated in  \cite{hampton_compressive_2015}. In order to optimize the sensing matrices, the authors adopt the minimization of coherence sensing matrices from several random samples of Legendre and Hermite polynomials. Using Monte Carlo Markov Chain (MCMC), the authors also derive the coherence of optimal-based random sampling points to employ $\ell_1$-minimization algorithm. Additionally, they also derived the coherence bound for a matrix constructed from those polynomials.

In contrast to all aforementioned results which rely on random samples from certain probability measures, in this work, we derive the coherence bound for deterministic sampling points and utilize properties of Wigner d-functions and their products and their representation in terms of Wigner 3j symbols. 


\subsection{Summary of Contributions}

In \cite{bangun_sensing_2020}, it was conjectured that the lower-bound on the mutual coherence is tight. In other words, the inner product of columns with equal orders is dominated by the inner product of two columns with zero order and highest degree. In this paper, we prove this conjecture and derive a set of related corollaries.  
The main contributions and some of the interesting conclusions of our paper are as follows:
\begin{itemize}
\item 
We show that  the product of Wigner D-functions can be written as a linear combination of single Legendre polynomials and Wigner 3j symbols. For equispaced sampling points on the elevation and using symmetry of Legendre polynomials, we show in Section \ref{sec2: prod} that only even degree Legendre polynomials contribute in the analysis, which in turn simplifies the problem formulation.

\item In Section \ref{sec3: sum_leg_express} and \ref{sec4: Ineq_3j_symbols}, we provide various inequalities and identities for sum of Legendre polynomials and Wigner 3j symbols. We establish monotonic properties of these terms as a function of degree and orders of Wigner D-functions. These results establish a certain ordering between inner products of the columns of the sensing matrix. Particularly we show that, under some conditions, the inner products have a specific order as a function of degrees and orders. As a corollary, we also present that the inner product of columns of equal orders are decreasing with orders and increasing with degree. The result can be used to obtain a  lower bound on
the mutual coherence. Proofs of main theorem, supporting lemmas, and propositions are given in Section \ref{maintheoremproof} and \ref{sec:proof_of_lemmas}

\item We numerically verify our results and show that our bound is larger than the Welch bound. Therefore, the desideratum of regular sampling pattern design should be this lower bound rather than the Welch bound. We also extend the sampling pattern design algorithm of \cite{bangun_sensing_2020} to a gradient-descent based algorithm and show that it can achieve the lower bound for spherical harmonics\footnote{
The codes used in this paper are available below:
\begin{center}
 {\fontfamily{lmss}\selectfont \url{github.com/bangunarya/boundwigner}}
\end{center}
}.
\end{itemize}

\subsection{Notation}
The vectors are denoted by bold small-cap letters $\mathbf{a},\mathbf{b},\dots$. Define $\N\defeq\{1,2,\dots\}$ and $\N_0\defeq \N\cup\{0\}$. The elevation, azimuth, and polarization angle are denoted by $\theta$, $\phi$, and $\chi$, respectively. The set $\{1, ..., m\}$ is denoted by $[m]$. $\overline{x}$ is the conjugate of $x$.

\section{Definitions and Problem Formulation} \label{sec2}
\subsection{Wigner D-functions}
The rotation group, denoted by $\mrm{SO}(3)$, consists of all possible rotations of the three-dimensional Euclidean space. Square integrable functions defined on the rotation group is the Hilbert space $L^2(\mrm{SO}(3))$ with the inner product of functions in this space $f,g \in L^2(\mrm{SO}(3))$ defined as
\[
\inner{f}g\defeq \int_{\SO(3)} f(\theta,\phi,\chi)\overline{g(\theta,\phi,\chi)}\d\nu(\theta,\phi,\chi).
\] 
Similar to the Fourier basis, which can be seen as the  eigenfunctions of Laplace operator on the unit circle, we can also derive eigenfunctions on the rotation group $\mathrm{SO(3)}$. These functions are called Wigner D-functions, sometimes also called generalized spherical harmonics. They are orthonormal basis for $L^2(\mrm{SO}(3))$. It can be written in terms of Euler angles $\theta\in[0,\pi]$, and $\phi , \chi\in[0,2\pi)$ as follows
\begin{equation}
\D l{k}{n}(\theta,\phi,\chi)= N_l e^{-\i k\phi} \mathrm{d}_l^{k,n}(\cos \theta)  e^{-\i n\chi} 
 \label{def:WigD}
\end{equation}
where $N_l=\sqrt{\frac{2l+1}{8\pi^2}}$ is the normalization factor to guarantee unit $L_2$-norm of Wigner D-functions. The function $\Wd l{k}{n}(\cos \theta)$ is the Wigner d-function of order $-l\leq k,n \leq 1$ and degree $l$ defined by
\begin{equation} \label{Wigner_d}
\Wd l{k}{n}(\cos \theta)\defeq \omega \sqrt{\gamma} \sin^{\xi} \bigg(\frac{\theta}{2}\bigg)\cos^{\lambda}\bigg(\frac{\theta}{2}\bigg) P_{\alpha}^{(\xi,\lambda)}(\cos \theta)
\end{equation}
where $\gamma=\frac{\alpha!(\alpha + \xi + \lambda)!}{(\alpha+\xi)!(\alpha+\lambda)!}$, $\xi=\card{k-n}$, $\lambda=\card{k+n}$, $\alpha=l-\big(\frac{\xi+\lambda}{2}\big)$ and 
\begin{equation*}
 \omega= \begin{cases} 
        1  & \text{if } n\geq k  \\
        (-1)^{n-k} & \text{if } n<k 
          \end{cases}.
\end{equation*}
The function $P_{\alpha}^{(\xi,\lambda)}$ is the Jacobi polynomial. 

Wigner D-functions are equal to complex spherical harmonics when the order $n$ is equal to zero, namely
\begin{equation}\label{sec2:eq gener_SH}
\D{l}{-k}{0}(\theta,\phi,0) = (-1)^k \sqrt{\frac{1}{2\pi}} \Y{l}{k}(\theta,\phi),
\end{equation}
where $\Y{l}{k}(\theta,\phi) \defeq  N^k_l P_l^{k}(\cos \theta) e^{\i k\phi}$. The term $N^k_l\defeq\sqrt{\frac{2l+1}{4\pi} \frac{(l-k)!}{(l+k)!}}$ is a normalization factor to ensure the function $\Y lk(\theta,\phi)$ has unit $L_2$-norm and $P_{l}^k (\cos \theta)$ is the associated Legendre polynomials. Wigner D-functions form an orthonormal basis with respect to the uniform measure on the rotation group, i.e., $\sin \theta \d\theta \mathrm{d}\phi \d\chi$.
\begin{equation} \label{orthogonality}
\small
\begin{aligned} 
 \int_{0}^{2\pi}\int_{0}^{2\pi}\int_{0}^{\pi}& \D{l}{k}{n}(\theta,\phi,\chi) \overline{\D {l'}{k'}{n'}(\theta,\phi,\chi)} \sin \theta \d\theta \d\phi \d\chi &=\delta_{ll'} \delta_{kk'} \delta_{nn'}.
\end{aligned}
\end{equation}
Similarly, Wigner d-functions are also orthogonal for different degree $l$ and fix order $k,n$:
\begin{equation} \label{inner_prod_Wigner}
\int_{0}^{\pi} \Wd{l}{k}{n}(\cos \theta) \Wd{l'}{k}{n}(\cos \theta) \sin \theta d\theta= \frac{2}{2l+1}\delta_{l l'}.
\end{equation}
The factor $\frac{2}{2l+1}$ can be used for normalization to get orthonormal pairs of functions.

As discussed earlier, the Wigner d-functions, or weighted Jacobi polynomials, are generalization of hypergeometric polynomials including associated Legendre polynomials, where the relationship between those polynomials can be expressed as  

\begin{equation} \label{legend_jacobi}
\Wd l{k}{0}(\cos \theta) = \sqrt{\frac{(l-k)!}{(l+k)!}}P_l^k(\cos \theta).
\end{equation}
The associated Legendre polynomials of degree $l$ and order $-l \leq k \leq l$ is given as
\begin{equation}
    P_l^k(\cos \theta) = (-1)^k  (\sin \theta)^{k} \frac{\d^k P_l(\cos \theta) }{\d(\cos \theta)^k},
\end{equation}
where $P_l(\cos \theta)$ is Legendre polynomial, and using Rodrigues formula it can be written as
\begin{equation} 
 P_l (\cos \theta) =   \frac{1}{2^l l!}\frac{\d^l}{\d\cos \theta^l} (\cos ^2\theta-1)^l.
\end{equation}
Similar to Wigner d-functions, Legendre polynomials are also orthonormal. The important properties of associated Legendre polynomials that are necessary in this article are also presented in Supplementary Material in Section \ref{App:sec2 hypergeo}.

\subsection{Problem Formulation}
In many signal processing applications, it is desirable to study properties of matrices constructed from samples of Wigner D-functions. A common example is reconstruction of band-limited functions on $\SO(3)$ from its samples. A function $g \in L^2(\SO(3))$ is band-limited with bandwidth $B$ if it is expressed in terms of Wigner D-functions of degree less than $B$:
\[
 g(\theta,\phi,\chi)=\sum_{l=0}^{B-1}\sum_{k=-l}^{l}\sum_{n=-l}^{l} \hat{g}_l^{k,n} \,\D l{k}{n}(\theta,\phi,\chi).
\]
Suppose that we take $m$ samples of this function at points $(\theta_p,\phi_p,\chi_p)$ for $p\in[m]$. The samples are put in the vector $\mbf y$, and the goal is to find the coefficients $\mbf g$. This is a linear inverse problem formulated by $\mbf y=\mbf A\mbf g$, with  $\mbf A$  given as
\begin{equation}
\small
\mat{A}= 
\begin{pmatrix}
  \D 0{0}{0}(\theta_1,\phi_1,\chi_1)&  \dots& \D {B-1}{B-1}{B-1}(\theta_1,\phi_1,\chi_1) \\
  \D 0{0}{0}(\theta_2,\phi_2,\chi_2) & \dots& \D {B-1}{B-1}{B-1}(\theta_2,\phi_2,\chi_2) \\
  \vdots &\ddots&\vdots\\
  \D 0{0}{0}(\theta_m,\phi_m,\chi_m) & \dots& \D {B-1}{B-1}{B-1}(\theta_m,\phi_m,\chi_m) 
\end{pmatrix}.
\label{ch3:eq sensingmatrix}
\end{equation}
 
For index column $q \in [N]$, we denotes the degree and orders of the respective Wigner D-function by $l(q)$, $k(q)$ and $n(q)$\footnote{The analytical description of these functions is not relevant for the rest of this paper. They are used mainly to ease the notation.}. The column dimension of this matrix is given by $N =\frac{B(2B-1)(2B+1)}{3}$. Using these functions, the elements of this matrix are given by
\begin{equation}
\mbf A=[{A}_{p,q}]_{p\in[m],q\in[N]}:\quad {A}_{p,q}=\mathrm{D}_{l{(q)}}^{k{(q)},n{(q)}}(\theta_p,\phi_p,\chi_p).
\label{ch3:eq sensMat}
\end{equation}
In compressed sensing, the sensing matrix $\mbf A$ with lower mutual coherence are more desirable for signal reconstruction \cite{foucart2013mathematical}. The mutual coherence, denoted by $\mu(\mbf A)$ is expressed as
\begin{equation} \label{coherence_Wigner}
\small
\begin{aligned}
& \underset{ 1 \leq r < q \leq N}{\text{max}}
 \,\,\card{\sum_{p=1}^{m}\frac{ \mathrm{D}_{l{(q)}}^{k{(q)},n{(q)}}(\theta_p,\phi_p,\chi_p)  \overline{\mathrm{D}_{l{(r)}}^{k{(r)},n{(r)}}(\theta_p,\phi_p,\chi_p)}}{\norm{\mathrm{D}_{l{(q)}}^{k{(q)},n{(q)}}(\boldsymbol\theta,\boldsymbol\phi,\boldsymbol\chi)}_2 \norm{\mathrm{D}_{l{(r)}}^{k{(r)},n{(r)}}(\boldsymbol\theta,\boldsymbol\phi,\boldsymbol\chi)}_2}},  
\end{aligned}
\end{equation}
where we adopt the following convention
\[\mathrm{D}_{l}^{k,n}(\boldsymbol\theta,\boldsymbol\phi,\boldsymbol\chi)\defeq
\begin{pmatrix}
 \mathrm{D}_{l}^{k,n}(\theta_1,\phi_1,\chi_1)\\
 \vdots\\
  \mathrm{D}_{l}^{k,n}(\theta_m,\phi_m,\chi_m)
\end{pmatrix}.
\]
For the rest of the paper, we focus mainly on the inner product between the samples. It can be numerically verified that the  $\ell_2$-norm of columns do not affect the  coherence value. {We will comment later on how these norms scale.}

Although the closed-form derivation of mutual coherence is in general difficult, the authors in \cite{bangun_sensing_2020} observed empirically that the mutual coherence of sensing matrices with equispaced sampling points on the elevation angle is indeed tightly bounded by a single term under certain assumptions. This is because the inner products of Wigner D-functions are ordered in a regular way as a function of their orders and degrees. In this paper, we provide theoretical supports for these observations. In other words, we provide simple analysis of mutual coherence for specific sampling patterns. Central to our analysis is a set of combinatorial identities about the sum and product of Wigner D-functions. We will focus on the equispaced sample on $\theta_p$ for $p\in[m]$, which are chosen such that 
\begin{equation}\label{ch5:eq equispaced}
  \cos \theta_{p} = \frac{2p-m-1}{m-1}.
\end{equation}
This means that \(-1=\cos \theta_{1}  <\cos \theta_2<\hdots< \cos \theta_{m-2} < \cos \theta_{m-1} < \cos \theta_{m}=1 \). There are multiple reasons for using this sampling pattern. First of all, this sampling pattern has been shown to be beneficial in spherical near-field antenna measurement \cite{bangun_coherence_2018,culotta-lopez_compressed_2018} where the robotic probe can acquire the electromagnetic field samples and move in the same distance. Second of all, this sampling pattern induces orthogonal columns in the sensing matrix between even and odd degree polynomials as discussed in \cite[Theorem 5]{bangun_sensing_2020}. Interestingly, fixing the sampling patterns on the elevation imposes a lower bound on the mutual coherence, which is tight in many cases. In this paper, we study the mutual coherence for this elevation sampling and arbitrary sampling patterns on $\phi$ and $\chi$.

\subsection{Product of Wigner D-functions} \label{sec2: prod}  
In the expression for  mutual coherence, product of Wigner D-functions appears constantly. This product appears also in the study of angular momentum in quantum mechanics and can be written as linear combination of a single Wigner D-functions with coefficients, called Wigner 3j symbols \cite{schulten_exact_1975,edmonds_angular_1974,rose_elementary_1995}. Using this representation, the discrete inner product of Wigner D-functions can be simplified as follows.
\begin{proposition}[\cite{bangun_sensing_2020}] \label{prop:CG_innerproduct}
Let $\D l{k}{n}(\theta ,\phi ,\chi )$ be the Wigner D-function with degree $l$ and orders $k,n$. Then we have:
\begin{equation}
\small
\begin{aligned}
 &\sum_{p=1}^m\overline{\D{l_1}{k_1}{n_1}(\theta_p,\phi_p,\chi_p)}  \D{l_2}{k_2}{n_2}(\theta_p,\phi_p,\chi_p)=   \\
 & C_{k_2,n_2}\sum_{\substack{\hat l=|l_2-l_1| }}^{l_1 + l_2} \sqrt{\frac{(2l_1 + 1)(2l_2 +1)(2\hat{l} + 1)}{8 \pi^2}}  \begin{pmatrix}
   l_1 & l_2 & \hat{l} \\
   -n_1 & n_2 & -\hat{n} 
  \end{pmatrix} \begin{pmatrix}
   l_1 & l_2 & \hat{l} \\
   -k_1 & k_2 & -\hat{k}
  \end{pmatrix} \sum_{p=1}^m \D {\hat{l}}{\hat{k}}{\hat{n}}(\theta_p,\phi_p,\chi_p) ,
\end{aligned}
\end{equation}
where $\hat{k} = k_2-k_1$, $\hat{n} = n_2 - n_1$ with the phase factor $ C_{k_2,n_2} = (-1)^{k_2 + n_2}$.
\end{proposition}
 
The parameters of Wigner 3j symbols are non-zero only under certain conditions known as the \textit{selection rules}. The selection rules state that Wigner 3j symbols $\begin{pmatrix}
   l_1 & l_2 & {l}_3 \\
   k_1 & k_3 & k_3
  \end{pmatrix} \in \R$
are non-zero if only if:
\begin{itemize}
\item The absolute value of $k_i$ does not exceed $l_i$, i.e., $-l_i \leq k_i \leq l_i$ for $i=1,2,3$
\item The summation of all $k_i$ should be zero: $k_1 + k_2 +k_3 = 0$.
\item Triangle inequality holds for $l_i$'s: $\card{l_1 - l_2} \leq l_3 \leq l_1 + l_2$.
\item The sum of all $l_i$'s should be an integer $l_1+l_2+l_3 \in \mathbb{N}$.
\item If $k_1 = k_2 = k_3 = 0$,  $l_1+l_2+l_3 \in \mathbb{N}$ should be an even integer.
\end{itemize}
There are other identities that will be useful for our derivations. For degrees $l_1,l_2$, orders $k_1,k_2,n_1,n_2$, and $\hat{k} = k_1 + k_2$, $\hat{n} = n_1 + n_2$, we obtain the following identities \cite{edmonds_angular_1974,rose_elementary_1995,schulten_exact_1975}. 
\begin{equation} \label{ch2:eq prop3j}
\begin{aligned} 
&\sum_{\hat{l}=\card{l_1-l_2}}^{l_1+l_2} (2\hat{l}+1)
\begin{pmatrix}
   l_1 & l_2 & \hat{l} \\
   k_1 & k_2 & -\hat{k} 
  \end{pmatrix}^2 = 1, \,\, \\
  &\sum_{\hat{l}=\card{l_1-l_2}}^{l_1+l_2} (2\hat{l} + 1) \begin{pmatrix}
   l_1 & l_2 & \hat{l}\\
   n_1 & n_2 & \hat{n} 
  \end{pmatrix}\begin{pmatrix}
   l_1 & l_2 & \hat{l}\\
   k_1 & k_2 & -\hat{k} 
  \end{pmatrix} = 0 , & \text{for} \,\, k_1 \neq n_1 \enskip \text{and} \enskip k_2 \neq n_2.
\end{aligned}
\end{equation}
Note that, from the selection rules, if $l_1 + l_2 + \hat{l}$ is an odd integer, then $\begin{pmatrix}    l_1 & l_2 & \hat{l}\\
   0 & 0 & 0
\end{pmatrix}$ is zero. Further properties and the exact expression of the Wigner 3j symbol will be included in the Supplementary Material in Section \ref{App:Property_Wigner3j}. 

It is trivial to derive  the product of same orders Wigner d-function, i.e.,  $k_1 = k_2 = k$ and $n_1 = n_2 = n$, as follows
\begin{equation}\label{product_Wigner_small}
\begin{aligned}
 &\card{\sum_{p=1}^m \Wd {l_1}{k }{n }(\cos\theta_p)   \Wd {l_2}{k}{n }(\cos\theta_p)}    
 =  \card{\sum_{\substack{\hat l=|l_2-l_1| }}^{l_1 + l_2} (2\hat{l} + 1) \begin{pmatrix}
   l_1 & l_2 & \hat{l} \\
   -n  & n  & 0 
  \end{pmatrix} \begin{pmatrix}
   l_1 & l_2 & \hat{l} \\
   -k  & k  & 0
  \end{pmatrix} \sum_{p=1}^mP_{\hat{l}}(\cos \theta_p)}
\end{aligned}
\end{equation}
An interesting conclusion of the above identities is that the sampling pattern affects the inner product through the sum of individual functions, for instance Legendre polynomials. 

\section{Main Results}\label{sec6: max_discrete}
The starting point of bounding the coherence is the following trivial inequality, which holds in full generality:
\begin{equation}
    \mu(\mbf A) \geq \max_{\substack{{l_1\neq l_2}\\{|k|,|n|\leq \min{(l_1,l_2)}}}}
\card{
\frac{
\sum_{p=1}^{m} \Wd {l_1}{k}{n}(\cos \theta_p)\Wd {l_2}{k}{n}(\cos \theta_p)
}{
\norm{\Wd {l_1}{k}{n}(\cos \boldsymbol\theta)}_2
\norm{\Wd {l_2}{k}{n}(\cos \boldsymbol\theta}_2
}
}.
\end{equation}
This is obtained by choosing the column with equal orders of $k$ and $n$, and using Definition \ref{def:WigD}.  In other words, for any sampling pattern, regardless of the choice of $\phi_p$ and $\chi_p$, the mutual coherence is lower bounded by merely choosing $\theta_p$. This indicates the sensitivity of mutual coherence to the sampling pattern on the elevation. The following theorem shows that the  maximum inner product in above expression has a simple solution for the sampling pattern \eqref{ch5:eq equispaced} if $m$ is moderately large. 

\begin{theorem}\label{ch5:thm prod}
Consider Wigner d-functions of degree $0 < l_1 < l_2  \leq B-1$ and orders $-\text{min}(l_1 ,l_2)\leq k,n  \leq  \text{min}(l_1,l_2)$, which are sampled according to \eqref{ch5:eq equispaced} with $m \geq \frac{(B + 2)^2}{10} + 1$. We have 
\begin{equation}
\begin{aligned}  
\max_{\substack{{l_1\neq l_2}\\{|k|,|n|\leq \min{(l_1,l_2)}}}}
&\card{\sum_{p=1}^{m} \Wd {l_1}{k}{n}(\cos \theta_p)\Wd {l_2}{k}{n}(\cos \theta_p)} =   \card{\sum_{p=1}^{m} {P}_{B-1}(\cos \theta_p)  {P}_{B-3}(\cos \theta_p)},
\end{aligned}
\end{equation}
where ${P}_l (\cos \theta)$ is the Legendre polynomial of degree $l$ with 
\[
 P_{l} (\cos \boldsymbol\theta)\defeq \paran{{P}_{l}(\cos \theta_1),\dots,{P}_{l}(\cos \theta_m)}^T.
\]
\end{theorem}


Intuitively, this theorem states that if one considers equispaced samples on the elevation, the maximum inner product occurs at the zero-order Wigner d-functions, which are Legendre polynomials. Additionally, we can obtain similar results for maximum discrete inner product of associated Legendre polynomials, since for $k=0$ or $n=0$ the Wigner d-functions are the associated Legendre polynomials. The result is given in the following corollary.
 
\begin{corollary}\label{ch5:cor 1}
Let  $\theta_p$'s for $p \in [m]$ be chosen as in \eqref{ch5:eq equispaced}. For $m \geq \frac{(B + 2)^2}{10} + 1$, we have
\begin{equation}
\begin{aligned}
\max_{\substack{{l_1\neq l_2}\\{|k|\leq \min{(l_1,l_2)}}}}\,\, &\left|\sum_{p=1}^{m} C_{l_1}^kC_{l_2}^kP^k_{l_1}(\cos \theta_p)P^k_{l_2}(\cos \theta_p)\right|  =&\left|\sum_{p=1}^{m} P_{B-1}(\cos \theta_p)P_{B-3}(\cos \theta_p)\right|,
\end{aligned}
\end{equation}
where $C^k_{l} = \sqrt{\frac{(l-k)!}{(l+k)!}}$ is the normalization factor.
\end{corollary}
Corollary \ref{ch5:cor 1} implies that  the maximum product of two associated Legendre polynomials for different degrees and same orders is also attained at degrees $B-1$ and $B-3$ and order $k=0$. 

 
\begin{remark}
A byproduct of Theorem \ref{ch5:thm prod} is that for a fixed number of measurement numbers $m$, the inner product of Wigner d-functions with degree less that $\sqrt{10 m}$ are  ordered. This surprising behavior in presented in numerical experiments. 
\end{remark}


The proof of the main result follows from a sequence of inequalities and identities. In the next sections, we provide some of them that are of independent interest. Interestingly, the well-ordered behaviour of the inner products of Wigner d-functions is linked to orders in the summand of \eqref{product_Wigner_small}.  
The proof leverages mostly classical inequalities and identities, e.g., Abel partial summation,  and the orders between Wigner d-functions. All the proofs for the following sections are given in Section \ref{sec:proof_of_lemmas}.

\section{Finite Sum of Legendre Polynomials}\label{sec3: sum_leg_express}

The starting point of the proof is to use \eqref{product_Wigner_small} for the product of Wigner d-functions. The proof strategy is based on establishing inequalities for each term in the sum \eqref{product_Wigner_small}, and then using them to bound the final sum. In this section, we provide a set of results for sum of Legendre polynomials.
The following lemma provides an identity for the sum of equispaced samples of Legendre polynomials.
\begin{lemma} \label{ch5:lemm sum_equi}
Suppose we have equispaced samples as in \eqref{ch5:eq equispaced}, then the sum of sampled Legendre polynomials for even degrees $l > 0$ is given by
\begin{equation*}
\sum_{p=1}^{m} P_{l}(\cos \theta_p) =1+\frac{l(l+1)}{6(m-1)}+ R_l(m),
\end{equation*}
where $\frac{l(l+1)}{6(m-1)} + R_l(m)$ is equal to $\sum_{k = 2 \atop k,\text{even}}^l \frac{(-1)^{\frac{k}{2} + 1}S_l^k}{(m-1)^{k-1}}$ and $S_l^k = \frac{\zeta(k) (l+k-1)! 4 }{(k-1)! (l-k+1)! (2\pi)^k}$ with $\zeta(k)$ is a zeta function. For odd degrees $l$, the sum is equal to zero.
\end{lemma}

Lemma \ref{ch5:lemm sum_equi} shows that we can simplify the summation of equispaced samples Legendre polynomials with respect to number of samples $m$, degree of polynomials $l$, and the residual $R_l(m)$.

\begin{remark}
If $l=0$, then the sum of equispaced samples Legendre polynomials is equal to $m$, since $P_0 (\cos \theta) = 1$. In addition, for $l=2$, the summation is equal to $1 + \frac{1}{m-1}$ and we do not have any residual. If we take a number of samples larger than the degree $l$, it is obvious that the summation converges to $1$.
\end{remark}

The residual $R_l(m)$ is important in the summation in Lemma \ref{ch5:lemm sum_equi}. We provide upper and lower bounds on the residual $R_l(m)$. The next proposition provides a bound on this summation. 


\begin{proposition} \label{ch5:prop residual}
Suppose we have $m  \geq \frac{(l+1)^2}{10} + 1$ with even degree $l \geq  4$. The residual $R_{l}(m)$ is therefore bounded by $-0.463 < R_l(m) < 0$.
\end{proposition}

The previous proposition shows that  the residual, conditioned on $m \geq \frac{(l+1)^2}{10} + 1$, is inside the interval, and most importantly is always negative. We present the numerical evaluation of this bound in Section \ref{sec6 NumEx}.  Using this property, not only  that the summation $\sum_{p=1}^{m} P_{l}(\cos \theta_p)$ is non-negative, but also it is monotonically increasing for an increasing even degree $l$.
\begin{lemma} \label{ch5:lemm increasing}
Let consider $m \geq \frac{(B + 2)^2}{10} + 1$, the sum of equispaced samples Legendre polynomials for even degrees  $0 \leq l \leq B-1$ is non-negative, i.e., $\sum_{p=1}^{m} P_{l}(\cos \theta_p) \geq 0$. Moreover, for an increasing sequence of even degrees $l$, i.e., $2< 4 < 6 < \hdots < B-1$, we have $
\sum_{p=1}^{m} P_{2}(\cos \theta_p) < \sum_{p=1}^{m} P_{4}(\cos \theta_p) < \sum_{p=1}^{m} P_{6}(\cos \theta_p) < \hdots < \sum_{p=1}^{m} P_{B-1}(\cos \theta_p).$
\end{lemma}

Lemma \ref{ch5:lemm sum_equi}, Proposition \ref{ch5:prop residual}, and Lemma \ref{ch5:lemm increasing} characterize the order of the sum of equispaced samples of Legendre polynomials. These properties are useful later to prove the main result in this paper. In the next section, we  show a similar ordering for other terms of expression in \eqref{product_Wigner_small}.


\section{Inequalities for Wigner 3j Symbols}\label{sec4: Ineq_3j_symbols}
To prove the main theorem, we establish that there is a similar order between Wigner 3j symbols. Some of these properties of Wigner 3j symbols are given in Section \ref{sec2} and in Supplementary Material in Section \ref{App:Property_Wigner3j}. In what follows, we will have some combinatorial identities and inequalities related to Wigner 3j symbols. Despite ample investigation of authors, it is not clear whether these results bear interesting implications for other areas particularly angular momentum analysis in quantum physics. For compressed sensing, however, these are quite interesting as they show that the sensing matrix from samples Wigner D-functions possesses a lot of structures and symmetries.
We start with the following lemma.
\begin{lemma}\label{ch5:lemm decreasingWigner3j}
Let assume we have degrees $0 \leq l_1 < l_2 \leq B-1$ and constant degrees $l_3$, then the following inequalities hold 
\begin{equation} 
\begin{aligned} 
\begin{pmatrix}
   l_1 & l_2 & l_3 \\
   0 & 0 & 0
\end{pmatrix}^2 
&\geq  
\begin{pmatrix}
   l_1+1 & l_2+1 & l_3 \\
   0 & 0 & 0
\end{pmatrix}^2,\\
 \begin{pmatrix}
   l_1 & l_2 & l_3 \\
   0 & 0 & 0
\end{pmatrix}^2  
   &\geq \begin{pmatrix}
   l_1+2 & l_2  & l_3 \\
   0 & 0 & 0
\end{pmatrix}^2.
\end{aligned}
\end{equation}
\end{lemma}

By reducing the indices one at a time, this result shows that the maximum of Wigner 3j symbols for zero order is achieved at zero degree. Using property in \eqref{ch2:eq prop3j}, one can directly obtain that the maximum is equal to $1$. The summation of Wigner 3j symbols for $k_1 = k_2$ is also essential to prove the main result and can be decomposed as the summation of odd and even degrees, as given below.
\begin{lemma} \label{ch5:lemm Wigner_3j_odd_even} 
Suppose that $l_1 \neq l_2 \in \mathbb{N}$ and $\card{l_1-l_2} \leq \hat{l} \leq l_1 + l_2$, then for $-\text{min}(l_1,l_2) \leq k\neq n \leq \text{min}(l_1,l_2)$, we have 
\begin{equation*} 
\begin{aligned}
 \sum_{\hat{l},\mrm{even}} (2\hat{l}+1)
\begin{pmatrix}
   l_1 & l_2 & \hat{l}\\
   -k & k & 0
  \end{pmatrix} \begin{pmatrix}
   l_1 & l_2 & \hat{l}\\
   -n & n & 0
  \end{pmatrix} =    \sum_{\hat{l},\mrm{odd}} (2\hat{l}+1)
\begin{pmatrix}
   l_1 & l_2 & \hat{l}\\
   -k & k & 0
  \end{pmatrix} \begin{pmatrix}
   l_1 & l_2 & \hat{l}\\
   -n & n & 0
  \end{pmatrix} &= 0.
\end{aligned}
\end{equation*}
Furthermore, for $k=n=\tau$ and $1 \leq\card{\tau} \leq \text{min}(l_1,l_2)$ we have
\begin{equation*} 
\small
\begin{aligned}
\sum_{\hat{l},\mrm{even}} (2\hat{l}+1)
\begin{pmatrix}
   l_1 & l_2 & \hat{l}\\
   -\tau & \tau & 0
  \end{pmatrix}^2 &= \sum_{\hat{l},\mrm{odd}} (2\hat{l}+1)
\begin{pmatrix}
   l_1 & l_2 & \hat{l}\\
   -\tau & \tau & 0
  \end{pmatrix}^2 = \frac{1}{2} 
\end{aligned}
\end{equation*}
\end{lemma}

\begin{remark}\label{ch5:remark sum_wigner_1}
In this lemma, we do not include the condition $ \tau   = 0 $ in the second property because it is obvious that we have $\begin{pmatrix}
   l_1 & l_2 & \hat{l}\\
   0 & 0 & 0
  \end{pmatrix}^2$ and from the selection rules in Section \ref{sec2} the sum $l_1 + l_2 + \hat{l}$ should be an even integer, which means that for even $\hat{l}$ and odd $l_1 + l_2$, the Wigner 3j symbols value is zero. From the orthogonal property of Wigner 3j symbols as in \eqref{ch2:eq prop3j}, we have 
$\sum_{\hat{l},\mrm{even}} (2\hat{l}+1)
\begin{pmatrix}
   l_1 & l_2 & \hat{l}\\
   0 & 0 & 0
  \end{pmatrix}^2 = 1$,
for $l_1 + l_2$ is even.
\end{remark}

The last result of this section is related to the product of Wigner 3j symbols and $\hat{l}(\hat{l} + 1)$, where $\card{l_1 - l_2} \leq \hat{l} \leq l_1 + l_2$ and $l_2 = l_1 + 2$. As discussed in Lemma \ref{ch5:lemm sum_equi}, the sum of equispaced samples Legendre polynomials can be expressed as $\sum_{p=1}^m P_l(\cos \theta_p) = 1 + \frac{l(l+1)}{6(m-1)} + R_l(m)$, where $R_l(m)$ is the residual. The following lemma gives an expression of the inner product between Wigner 3j symbols and $l(l+1)$.
\begin{lemma} \label{ch5:lemm Wigner_3j_integer}
Let consider Wigner 3j symbols $\begin{pmatrix}
   l_1 & l_1 + 2 & \hat{l} \\
   0 & 0 & 0 
  \end{pmatrix}^2$, where the degree $l_2=l_1+2$ and $2 \leq \hat{l} \leq 2l_1 + 2$. Hence, we have the following equality 
\begin{equation*}
\small
\begin{aligned}
&\sum_{\hat{l}= 2 \atop \hat{l},\mrm{even}}^{2l_1 + 2}  (2\hat{l} + 1) \begin{pmatrix}
   l_1 & l_1 + 2 & \hat{l} \\
   0 & 0 & 0 
  \end{pmatrix}^2(\hat{l}^2 + \hat{l}) = 2 + 2(l_1+2)(l_1 +1).
\end{aligned}
\end{equation*}
\end{lemma}

Since we consider degree $0 \leq l_1 < l_2 \leq B-1$ and order $k=n=0$, the previous lemma has a direct implication for $l_1 = B-3$ and $l_2 =B-1$, as given in the following corollary. 
\begin{corollary}\label{ch5:cor max_B}
Suppose $l_1 = B-3$ and $l_2 = B-1$, then we have 
\begin{equation*}
\footnotesize
\begin{aligned}
\sum_{\hat{l}=2 \atop \hat{l},\mrm{even}}^{2B -4}  (2\hat{l} + 1) \begin{pmatrix}
   B-3 & B-1 & \hat{l} \\
   0 & 0 & 0 
  \end{pmatrix}^2(\hat{l}^2 + \hat{l}) = 2 + 2(B-1)(B-2).
\end{aligned}
\end{equation*}
\end{corollary}

This corollary is important to determine the product of Wigner 3j symbols and the sum of equispaced samples of Legendre polynomials. Since the latter can be expressed by $\hat{l}(\hat{l} + 1)$, as shown in Lemma \ref{ch5:lemm sum_equi} , we can directly apply Lemma \ref{ch5:lemm Wigner_3j_integer} to estimate the product.

\section{Experimental Results}\label{sec6 NumEx}

In what follows, we conduct a series of experiments to verify some of the results in the paper, as well as applications in compressed sensing. 
\subsection{Numerical verification of theoretical results}
In Lemma \ref{ch5:lemm sum_equi}, we can express the sum of equispaced samples Legendre polynomials as $\sum_{p=1}^{m} P_{l}(\cos \theta_p) =1+\frac{l(l+1)}{6(m-1)}+ R_l(m)$, where $R_l(m)$ is the residual with interval $-0.463 < R_l(m) < 0$ by considering $m \geq \frac{(l + 1)^2}{10} + 1$, as in Proposition \ref{ch5:prop residual}. 
\begin{figure}[!htb]
\centering
     \scalebox{0.5}{
\begin{tikzpicture}

\begin{axis}[%
width=5in,
height=2.5in,
at={(2.611in,1.741in)},
scale only axis,
colorbar,
colorbar style={ylabel={}},
colormap={mymap}{[1pt]
  rgb(0pt)=(0,0,0.5);
  rgb(22pt)=(0,0,1);
  rgb(25pt)=(0,0,1);
  rgb(68pt)=(0,0.86,1);
  rgb(70pt)=(0,0.9,0.967741935483871);
  rgb(75pt)=(0.0806451612903226,1,0.887096774193548);
  rgb(128pt)=(0.935483870967742,1,0.0322580645161291);
  rgb(130pt)=(0.967741935483871,0.962962962962963,0);
  rgb(132pt)=(1,0.925925925925926,0);
  rgb(178pt)=(1,0.0740740740740741,0);
  rgb(182pt)=(0.909090909090909,0,0);
  rgb(200pt)=(0.5,0,0)
},
legend cell align={left},
legend style={fill opacity=0.8, draw opacity=1, text opacity=1, draw=white!80!black},
point meta max=1.56519242011655e-12,
point meta min=-0.51,
tick align=outside,
tick pos=left,
x grid style={white!69.0196078431373!black},
xlabel={Degree (i)},
xmin=2, xmax=200,
xtick style={color=black},
y grid style={white!69.0196078431373!black},
ylabel={Samples (m)},
ymin=3, ymax=3003,
ytick style={color=black}
]
\addplot[forget plot] graphics [includegraphics cmd=\pgfimage,xmin=0, xmax=200, ymin=0, ymax=3001] {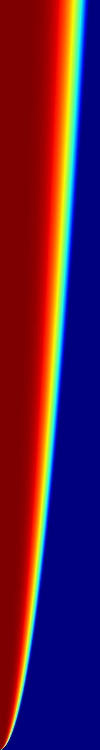};
\addplot [thick, line width = 2.0pt, black]
table {%
2 1.4
4 2.6
6 4.6
8 7.4
10 11
12 15.4
14 20.6
16 26.6
18 33.4
20 41
22 49.4
24 58.6
26 68.6
28 79.4
30 91
32 103.4
34 116.6
36 130.6
38 145.4
40 161
42 177.4
44 194.6
46 212.6
48 231.4
50 251
52 271.4
54 292.6
56 314.6
58 337.4
60 361
62 385.4
64 410.6
66 436.6
68 463.4
70 491
72 519.4
74 548.6
76 578.6
78 609.4
80 641
82 673.4
84 706.6
86 740.6
88 775.4
90 811
92 847.4
94 884.6
96 922.6
98 961.4
100 1001
102 1041.4
104 1082.6
106 1124.6
108 1167.4
110 1211
112 1255.4
114 1300.6
116 1346.6
118 1393.4
120 1441
122 1489.4
124 1538.6
126 1588.6
128 1639.4
130 1691
132 1743.4
134 1796.6
136 1850.6
138 1905.4
140 1961
142 2017.4
144 2074.6
146 2132.6
148 2191.4
150 2251
152 2311.4
154 2372.6
156 2434.6
158 2497.4
160 2561
162 2625.4
164 2690.6
166 2756.6
168 2823.4
170 2891
172 2959.4
174 3028.6
176 3098.6
178 3169.4
180 3241
182 3313.4
184 3386.6
186 3460.6
188 3535.4
190 3611
192 3687.4
194 3764.6
196 3842.6
198 3921.4
200 4001
};
\addlegendentry{$m = \frac{(l + 1)^2}{10} + 1$}
\end{axis}

\end{tikzpicture}}  

\caption{Residual error of sum of equispaced samples Legendre polynomials from $1+\frac{l(l+1)}{6(m+1)}$.}
    \label{ch5:Fig Residual_sum}
\end{figure}
The numerical evaluation of this proposition is presented in Figure \ref{ch5:Fig Residual_sum}, where it can be seen that by considering $m = \frac{(l + 1)^2}{10} + 1$, represented with the black line, the residual is restricted within the interval by the blue and red colors, respectively. In other words, the obtained constants are indeed tight in Proposition \ref{ch5:prop residual}.

Next, we numerically evaluate Theorem \ref{ch5:thm prod}. In  Figure \ref{MaxProd2D}, we show for which pairs of $(m,B)$, the identity of Theorem \ref{ch5:thm prod} holds using the color red. 
We have furthermore included a black line indicating the number of samples as in Theorem \ref{ch5:thm prod} for different $B$. It can be seen for small $B$, the condition is tight. However, it seems that it can be improved for larger values of $B$. {In Figure \ref{MaxProd2D}, the numerical experiments are performed by considering the normalization with respect to the $\ell_2$-norm. Thereby, we can numerically verify that the normalization does not affect the inequality in Theorem  \ref{ch5:thm prod}}.
\begin{figure}[!htb]
\centering
     \scalebox{0.5}{
\begin{tikzpicture}

\begin{axis}[%
width=5in,
height=2.5in,
at={(2.611in,1.741in)},
scale only axis,
colormap={mymap}{[1pt]
  rgb(0pt)=(0,0,0.5);
  rgb(22pt)=(0,0,1);
  rgb(25pt)=(0,0,1);
  rgb(68pt)=(0,0.86,1);
  rgb(70pt)=(0,0.9,0.967741935483871);
  rgb(75pt)=(0.0806451612903226,1,0.887096774193548);
  rgb(128pt)=(0.935483870967742,1,0.0322580645161291);
  rgb(130pt)=(0.967741935483871,0.962962962962963,0);
  rgb(132pt)=(1,0.925925925925926,0);
  rgb(178pt)=(1,0.0740740740740741,0);
  rgb(182pt)=(0.909090909090909,0,0);
  rgb(200pt)=(0.5,0,0)
},
legend cell align={left},
legend style={fill opacity=0.8, draw opacity=1, text opacity=1, draw=white!80!black},
point meta max=1,
point meta min=0,
tick align=outside,
tick pos=left,
x grid style={white!69.0196078431373!black},
xlabel={Bandwidth (B)},
xmin=4, xmax=46,
xtick style={color=black},
y grid style={white!69.0196078431373!black},
ylabel={Samples (m)},
ymin=7, ymax=299,
ytick style={color=black}
]
\addplot[forget plot] graphics [includegraphics cmd=\pgfimage,xmin=4, xmax=46, ymin=7, ymax=299] {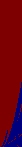};
\addplot [thick, line width = 4.0pt, black]
table {%
4 4.6
6 7.4
8 11
10 15.4
12 20.6
14 26.6
16 33.4
18 41
20 49.4
22 58.6
24 68.6
26 79.4
28 91
30 103.4
32 116.6
34 130.6
36 145.4
38 161
40 177.4
42 194.6
44 212.6
46 231.4
};
\addlegendentry{$m = \frac{(B + 2)^2}{10} + 1$}
\end{axis}

\end{tikzpicture}}  

\caption{Numerical verification of Theorem \ref{ch5:thm prod}}
\label{MaxProd2D}
\end{figure}

\subsection{Comparison with Welch bound and designing sampling patterns}

Our theoretical result provides a lower bound on the mutual coherence for equispaced samples on the elevation angle. Specifically, we have $\cos \theta_p = \frac{2p - m - 1}{m-1}$ for $p \in [m]$. It is interesting to see if the bound improves on previously existing bounds, like Welch bound, and if the bound can be somehow achieved.

In order to design sensing matrices from spherical harmonics and Wigner D-function, we choose points on azimuth and polarization angles $\phi,\chi \in [0, 2\pi)$ using the optimization problem below:
\begin{equation*}
\begin{aligned}
& \underset{\bs \phi, \bs \chi}{\text{minimize}}
& & \underset{ 1 \leq r < q \leq N}{\text{max}} \card{f_{q,r}(\bs \theta,\bs \phi,\bs \chi)}\\
& \text{subject to}
& & \phi_p,\chi_p \in [0,2\pi] \quad  \text{for} \quad p \in [m]
\end{aligned}
\end{equation*}
where $f_{q,r}(\bs \theta,\bs \phi,\bs \chi)$ is given as
$$  \frac{\sum_{p=1}^{m} \mathrm{D}_{l{(q)}}^{k{(q)},n{(q)}}(\theta_p,\phi_p,\chi_p)  \overline{\mathrm{D}_{l{(r)}}^{k{(r)},n{(r)}}(\theta_p,\phi_p,\chi_p)}}{\norm{\mathrm{D}_{l{(q)}}^{k{(q)},n{(q)}}(\boldsymbol\theta,\boldsymbol\phi,\boldsymbol\chi)}_2 \norm{\mathrm{D}_{l{(r)}}^{k{(r)},n{(r)}}(\boldsymbol\theta,\boldsymbol\phi,\boldsymbol\chi)}_2}. $$

For spherical harmonics, one can generate the problem from a relation in \eqref{sec2:eq gener_SH} without constraint on the polarization angle $\chi$ and order $n$. 

This optimization problem is a challenging min-max problem with non-smooth objective function and generally non-convex. In \cite{bangun_sensing_2020}, we have used a search-based method for optimization, which turns out to be difficult to tune and more time-consuming. We introduce a relaxation of the above problem and use gradient-descent based algorithms for optimizing it. Derivative of spherical harmonics and Wigner D-functions are given in Section \ref{Deriv}.

Using property of $\ell_p$-norm, we can write the objective function as
$$\underset{\bs \phi, \bs \chi}{\text{min}}\,\, \underset{p \rightarrow \infty}{\lim} \left(\underset{1 \leq r < q \leq N}{\sum} \card{f_{q,r}(\bs \theta,\bs \phi,\bs \chi)}^{p}\right)^{1/p}.$$ One can choose large enough $p$ and calculate the gradient. 
Therefore, we use gradient descent algorithms to solve this problem, as given in Algorithm \ref{algo_ps}.

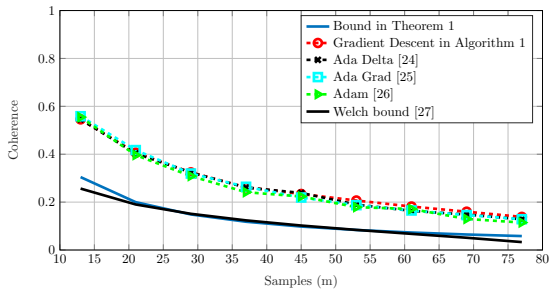
\begin{figure}[!htb]
\centering
     \scalebox{0.5}{
%
%
\definecolor{mycolor1}{rgb}{0.00000,0.44700,0.74100}%
\definecolor{mycolor2}{rgb}{0.00000,1.00000,1.00000}%
\definecolor{mycolor3}{rgb}{1.00000,0.00000,1.00000}%
\begin{tikzpicture}

\begin{axis}[%
width=5in,
height=2.5in,
at={(1.672in,0.489in)},
scale only axis,
xmin=10,
xmax=80,
xlabel style={font=\color{white!15!black}},
xlabel={Samples (m)},
ymin=0,
ymax=1,
ylabel style={font=\color{white!15!black}},
ylabel={Coherence},
axis background/.style={fill=white},
xmajorgrids,
ymajorgrids,
legend style={legend cell align=left, align=left, draw=white!15!black}
]
\addplot [color=mycolor1,line width=2.0pt ]
  table[row sep=crcr]{%
13	0.304467341439753\\
21	0.199796472131812\\
29	0.148436219218273\\
37	0.118008484060507\\
45	0.0979065566449628\\
53	0.0836448225552256\\
61	0.0730041797034076\\
69	0.064762306196862\\
77	0.0581909239467795\\
};
\addlegendentry{Bound in Theorem \ref{ch5:thm prod}}

\addplot [color=red, dashed,line width=2.0pt, mark size=3.0pt, mark=o, mark options={solid, red}]
  table[row sep=crcr]{%
13	0.545558715234713\\
21	0.406479149574366\\
29	0.324250233357159\\
37	0.262056154714913\\
45	0.233887490235098\\
53	0.206703616611632\\
61	0.181509670623068\\
69	0.159991563817843\\
77	0.13829361588879\\
};
\addlegendentry{Gradient Descent in Algorithm \ref{algo_ps}}

\addplot [color=black, dashed,line width=2.0pt, mark size=3.0pt, mark=x, mark options={solid, black}]
  table[row sep=crcr]{%
13	0.545266379740113\\
21	0.403103536547602\\
29	0.323691007122223\\
37	0.260935706002292\\
45	0.238352398046734\\
53	0.18970365421016\\
61	0.163538598020524\\
69	0.148146527560266\\
77	0.128152459766674\\
};
\addlegendentry{Ada Delta \cite{zeiler_adadelta_2012}}

\addplot [color=mycolor2, dashed,line width=2.0pt, mark size=3.0pt, mark=square, mark options={solid, mycolor2}]
  table[row sep=crcr]{%
13	0.558641919598301\\
21	0.417230856607857\\
29	0.319364645829167\\
37	0.263962475769476\\
45	0.219215014203741\\
53	0.187337586167699\\
61	0.165065110884299\\
69	0.145797547864673\\
77	0.128583340350047\\
};
\addlegendentry{Ada Grad \cite{duchi_adaptive_2011}}

\addplot [color=green, dashed, mark=triangle,line width=2.0pt, mark size=3.0pt, mark options={solid, rotate=270, green}]
  table[row sep=crcr]{%
13	0.554800316004069\\
21	0.396358517264571\\
29	0.30871618029259\\
37	0.241362339665753\\
45	0.223548230053292\\
53	0.178846410767328\\
61	0.171370145818752\\
69	0.129674452566332\\
77	0.114632141664637\\
};
\addlegendentry{Adam \cite{kingma_adam_2017}}

\addplot [color=black,line width=2.0pt]
  table[row sep=crcr]{%
13	0.256518358429386\\
21	0.190117275157343\\
29	0.151162233288431\\
37	0.123711200332038\\
45	0.102184964981513\\
53	0.0839467459114694\\
61	0.0674000089671322\\
69	0.0511778501066631\\
77	0.0330951696160748\\
};
\addlegendentry{Welch bound \cite{welch_lower_1974}}
								
\end{axis}
\end{tikzpicture}%
}  

\caption{Coherence of Wigner D-functions sensing matrix ($B = 4$)}
\label{Fig opt_Wigner}
\end{figure}
Figure \ref{Fig opt_Wigner} shows the coherence of a sensing matrix from Wigner D-functions using sampling points generated from several stochastic gradient algorithms. Although there are no sampling points that reach the lower bound in Theorem \ref{ch5:thm prod}, it can be seen that Adam algorithm \cite{kingma_adam_2017} yields the best sampling points. In this work, we also compare with several well-known stochastic gradient descent, as given in \cite{kingma_adam_2017,duchi_adaptive_2011,zeiler_adadelta_2012}. Bandwidth of Wigner D-functions is $B=4$ or the column dimension $N = 84$.
 \begin{algorithm}
    \caption{{Gradient Descent}}\label{algo_ps}
    \begin{algorithmic}
      \scriptsize
      \STATE  Initialization : $\bs \phi_0, \bs \chi_0$ uniformly random $[0,2\pi)$, step size $\eta= 0.5$, $\epsilon = 10^{-4}$, bound $\mu_{LB}$ from Theorem \ref{ch5:thm prod}, $\ell_{p}$-norm with $p = 8$ 
       
    \FOR  { $i = 1, \hdots, i_{\text{max}}$ or $\card{\mu_{LB} - \mu{(\mat A)}} \leq \epsilon$ }
    \STATE{ \begin{equation*}
    \begin{aligned}\bs \phi_i &= \bs \phi_{i-1} - \eta \nabla_{\bs \phi} \left(\underset{ 1 \leq r < q \leq N}{\sum} \card{{f_{q,r}(\bs \theta,\bs \phi,\bs \chi)}}^{p}\right)^{1/p}\\
    \bs \chi_i &= \bs \chi_{i-1} - \eta \nabla_{\bs \chi} \left(\underset{ 1 \leq r < q \leq N}{\sum} \card{{f_{q,r}(\bs \theta,\bs \phi,\bs \chi)}}^{p}\right)^{1/{p}}\\
    \mu(\mat A) &=  \underset{ 1 \leq r < q \leq N}{\text{max}} \card{f_{q,r}(\bs \theta,\bs \phi_i,\bs \chi_i)}
    \end{aligned}
    \end{equation*}}
    
      \ENDFOR   
    \end{algorithmic}
  \end{algorithm} 

Algorithm \ref{algo_ps} can be tailored for spherical harmonics, by only considering parameter on azimuth $\bs \phi$ and using relation between Wigner D-functions and spherical harmonics in \eqref{sec2:eq gener_SH}. {It can be seen that most of the gradient descent based algorithms converge to the lower bound for spherical harmonics, as shown in Figure \ref{Fig opt_SH}}. Therefore, we can provide sampling points on the sphere with mutual coherence that can achieve the  the lower bound in Theorem \ref{ch5:thm prod}. In this case, the spherical harmonics are generated with the bandwidth $B=10$ or equivalently we have the column dimension $N = 100$. 

\begin{figure}[!htb]
\centering
     \scalebox{0.5}{
%
%
\definecolor{mycolor1}{rgb}{0.00000,0.44700,0.74100}%
\definecolor{mycolor2}{rgb}{0.00000,1.00000,1.00000}%
\definecolor{mycolor3}{rgb}{1.00000,0.00000,1.00000}%
\begin{tikzpicture}

\begin{axis}[%
width=5in,
height=2.5in,
at={(1.672in,0.489in)},
scale only axis,
xmin=10,
xmax=100,
xlabel style={font=\color{white!15!black}},
xlabel={Samples (m)},
ymin=0,
ymax=1,
ylabel style={font=\color{white!15!black}},
ylabel={Coherence},
axis background/.style={fill=white},
xmajorgrids,
ymajorgrids,
legend style={legend cell align=left, align=left, draw=white!15!black}
]
\addplot [color=mycolor1, line width=2.0pt]
  table[row sep=crcr]{%
17	0.678688585148102\\
25	0.562539483279066\\
33	0.468247493229728\\
41	0.396366496625075\\
49	0.34134648342864\\
57	0.298524542357621\\
65	0.264555317970759\\
73	0.237109582915589\\
81	0.214560787501566\\
89	0.195757082980803\\
97	0.1798683518752\\
};
\addlegendentry{Bound in Theorem \ref{ch5:thm prod}}

\addplot [color=red,line width=2.0pt, mark size=3.0pt, dashed, mark=o, mark options={solid, red}]
  table[row sep=crcr]{%
17	0.678688585148103\\
25	0.562539483279067\\
33	0.468247493229728\\
41	0.396366496625076\\
49	0.34134648342864\\
57	0.29852454235762\\
65	0.264555317970759\\
73	0.237109582915589\\
81	0.214560787501566\\
89	0.195757082980803\\
97	0.1798683518752\\
};
\addlegendentry{Gradient Descent in Algorithm \ref{algo_ps}}

\addplot [color=black,line width=2.0pt, mark size=3.0pt, dashed, mark=x, mark options={solid, black}]
  table[row sep=crcr]{%
17	0.678688585148103\\
25	0.562539483279067\\
33	0.468247493229728\\
41	0.396366496625076\\
49	0.34134648342864\\
57	0.29852454235762\\
65	0.264555317970759\\
73	0.237109582915589\\
81	0.214560787501566\\
89	0.195757082980803\\
97	0.1798683518752\\
};
\addlegendentry{Ada Delta \cite{zeiler_adadelta_2012}}

\addplot [color=mycolor2,line width=2.0pt, mark size=3.0pt, dashed, mark=square, mark options={solid, mycolor2}]
  table[row sep=crcr]{%
17	0.678688585148103\\
25	0.562539483279067\\
33	0.468247493229728\\
41	0.396366496625076\\
49	0.34134648342864\\
57	0.29852454235762\\
65	0.264555317970759\\
73	0.237109582915589\\
81	0.214560787501566\\
89	0.195757082980803\\
97	0.1798683518752\\
};
\addlegendentry{Ada Grad \cite{duchi_adaptive_2011}}
\addplot [color=green,line width=2.0pt, mark size=3.0pt, dashed, mark=triangle, mark options={solid, rotate=270, green}]
  table[row sep=crcr]{%
17	0.678688585148103\\
25	0.562539483279067\\
33	0.468247493229728\\
41	0.396366496625076\\
49	0.34134648342864\\
57	0.29852454235762\\
65	0.264555317970759\\
73	0.237109582915589\\
81	0.224648115265434\\
89	0.195757082980803\\
97	0.193614425658094\\
};
\addlegendentry{Adam \cite{kingma_adam_2017}}

\addplot [color=black,line width=2.0pt]
  table[row sep=crcr]{%
17	0.222073628275669\\
25	0.174077655955698\\
33	0.143206534128867\\
41	0.120563675713100\\
49	0.102534366080753\\
57	0.0872929564835514\\
65	0.0737496131447850\\
73	0.0611227456628046\\
81	0.0486762030127979\\
89	0.0353332626668787\\
97	0.0176749080410067\\
};
\addlegendentry{Welch bound \cite{welch_lower_1974}}
\end{axis}

\end{tikzpicture}%
}  

\caption{Coherence of spherical harmonics sensing matrix ($B = 10)$ }
\label{Fig opt_SH}
\end{figure}
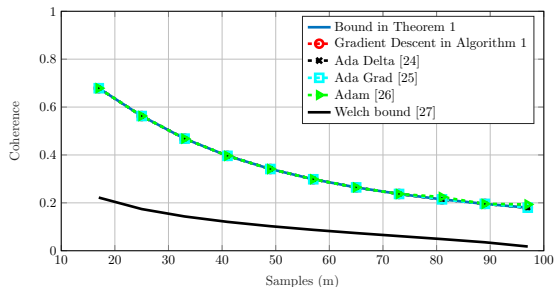
\section{Conclusions and Future Works} \label{Conclusion}
We have established a coherence bound of sensing matrices for Wigner D-functions on regular grids. This result also holds for spherical harmonics, which is a special case for Wigner D-functions. Estimating coherence involves non-trivial and complicated product of two Wigner D-functions for all combination of degrees and orders, yielding an obstacle to derive a simple and compact formulation of the coherence bound. 
Using the tools in the area of angular momentum in quantum mechanics disentangles this problem and represents the product as a linear combination of single Wigner D-functions and angular momentum coefficients, so called Wigner 3j symbols. In this paper, we derive some interesting properties of these coefficients and finite summation of Legendre polynomials to obtain the coherence bound. We also provided numerical experiments in order to verify the tightness of this bound. 
For practical application, it is also necessary to provide sampling points on the sphere and the rotation group that can achieve the coherence bound. We have shown that, for spherical harmonics, one can generate points to achieve this bound by using class of gradient descent algorithms. Convergence analysis of these algorithms and construction of deterministic sampling points to achieve this bound will be relegated for future works.
\section*{Acknowledgment}
This work is funded by DFG project (CoSSTra-MA1184 $|$ 31-1).

\section{Proofs of Main Theorem}\label{maintheoremproof}
\begin{proof}
In order to prove the main theorem, it is enough to show that for degree $0 \leq l_1 < l_2 \leq B-1$ and orders $-\text{min}(l_1,l_2) \leq k,n \leq  \text{min}(l_1,l_2)$, the following inequality holds
\begin{equation}
\small
\card{\sum_{p=1}^{m} {P}_{B-1}(\cos \theta_p)  {P}_{B-3}(\cos \theta_p)}\geq\card{\sum_{p=1}^{m} \Wd {l_1}{k}{n}(\cos \theta_p)\Wd {l_2}{k}{n}(\cos \theta_p)}.
\end{equation}
We can expand the product of Wigner d-functions in the right hand side as 
$$\sum_{\substack{\hat l=|l_2-l_1|\\\hat{l},\text{even}}}^{l_1+l_2} (2\hat{l} + 1) 
  \begin{pmatrix}
   l_1&l_2& \hat{l} \\
   -n & n & 0 
  \end{pmatrix}\begin{pmatrix}
   l_1 & l_2 & \hat{l} \\
   -k & k & 0 
  \end{pmatrix}\sum_{p=1}^m P_{\hat{l}} (\cos \theta_p).$$ Additionally, the product of Legendre polynomials in left hand side can be written as $$\sum_{\substack{\hat l=2\\\hat{l},\text{even}}}^{2B-4} (2\hat{l} + 1) 
  \begin{pmatrix}
   B-3&B-1& \hat{l} \\
   0 & 0 & 0 
  \end{pmatrix}^2\sum_{p=1}^m P_{\hat{l}} (\cos \theta_p).$$
The strategies are divided into two parts, which are for order $k = n = 0$ and for other conditions of order $k,n$.

Let consider the first case, for zero order $k = n = 0$.
From \eqref{legend_jacobi}, we know that $\Wd {l}{0}{0}(\cos \theta) = P_l(\cos \theta)$. Thus, it is equivalent to prove the maximum product of two Legendre polynomials is attained at $l_1=B-3$ and $l_2 = B-1$, i.e., $\underset{l_1 \neq l_2}{\max}\enskip \card{\sum_{p=1}^{m} P_{l_1}(\cos \theta_p)P_{l_2}(\cos \theta_p)} =\card{\sum_{p=1}^{m} P_{B-1}(\cos \theta_p)P_{B-3}(\cos \theta_p)}.$ We show that for an even $l_1+l_2$, if $(l_1,l_2)$ is increased to either $(l_1+1,l_2+1)$ or $(l_1+2,l_2)$, the sum of the product of two Legendre polynomials increases. It is enough to consider these two situations since from any pair $(l_1,l_2)$, one can use a sequence of inequalities to arrive at $(B-3,B-1)$. For an odd $l_1+l_2$, it is implied that the inner product is zero.


We use the product in \eqref{product_Wigner_small} for $k=n=0$ to get the representation of $\sum_{p=1}^mP_{l_1}(\cos \theta_p) P_{l_2}(\cos \theta_p)$ as
\begin{equation}
\begin{aligned}
\sum_{\substack{\hat l=|l_2-l_1|\\\hat{l}, \text{even}}}^{l_1 + l_2} (2\hat{l} + 1) \begin{pmatrix}
   l_1 & l_2 & \hat{l} \\
   0 & 0 & 0 
  \end{pmatrix}^2 \sum_{p=1}^{m} P_{\hat{l}} (\cos \theta_p)=\sum_{\substack{\hat{l}=|l_2-l_1|\\\hat l, \text{even}}}^{l_2+l_1} a_{\hat{l}}c_{\hat{l}},
\end{aligned}
\end{equation}
where from Lemma \ref{ch5:lemm increasing}, $c_{\hat{l}}=\sum_{p=1}^{m}P_{\hat{l}}  (\cos \theta_p)$ is non-negative and increasing for even values of $\hat l \geq 2$. The corresponding Wigner 3j symbols for $(l_1,l_2,\hat l)$ is denoted by $a_{\hat{l}}=(2\hat{l} + 1) \begin{pmatrix}
   l_1 & l_2 & \hat{l} \\
   0 & 0 & 0
  \end{pmatrix}^2$. Suppose we have $b_{\hat{l}}=(2\hat{l} + 1) \begin{pmatrix}
   l_1+1 & l_2+1 & \hat{l} \\
   0 & 0 & 0
  \end{pmatrix}^2$ for $(l_1 + 1,l_2+1)$  or $b_{\hat{l}}=(2\hat{l} + 1) \begin{pmatrix}
   l_1+2 & l_2 & \hat{l} \\
   0 & 0 & 0
  \end{pmatrix}^2$ for $(l_1 + 2,l_2)$. From Lemma \ref{ch5:lemm decreasingWigner3j} we have $a_{\hat l}\geq b_{\hat l}$ for both cases and additionally $\sum_{\hat l}a_{\hat{l}} =1$ and $\sum_{\hat l}b_{\hat{l}}=1$, as pointed out in \eqref{ch2:eq prop3j}. By using Abel's partial summation formula as stated in \eqref{App: Abelsum}, we have
 \begin{equation}
 \begin{aligned}
  \sum_{\substack{\hat l=|l_2-l_1|\\\hat l, \text{even}}}^{l_1+l_2} a_{\hat{l}} c_{\hat{l}}&=\sum_{\substack{\hat l=|l_2-l_1|\\\hat l, \text{even}}}^{l_2+l_1} A_{\hat l}(c_{\hat l} - c_{\hat l+2})+c_{l_1 + l_2 + 2}, \\
 \sum_{\substack{\hat l=|l_2-l_1|\\\hat l, \text{even}}}^{l_1+l_2+2} b_{\hat{l}} c_{\hat{l}}&=\sum_{\substack{\hat l=|l_2-l_1|\\\hat l, \text{even}}}^{l_2+l_1} B_{\hat l}(c_{\hat l} - c_{\hat l+2})+c_{l_1 + l_2 + 2},
\end{aligned} 
 \end{equation}  
 where $A_{\hat{l}} = \sum_{j, \text{even} =\card{l_1 -l_2}}^{\hat{l}} a_{j}$ and $B_{\hat l}$ are defined accordingly for $b_{\hat l}$. Since $A_{\hat l}\geq B_{\hat l}$ and $c_{\hat l}$ is increasing, it is clear that $ \sum_{\hat l} a_{\hat{l}} c_{\hat{l}}\leq  \sum_{\hat l} b_{\hat{l}} c_{\hat{l}}$, which establishes desired result by increasing the degrees until their reach $l_1 = B-3$ and $l_2 = B-1$.
 
For $k \neq n$ and $k = n \neq 0$, we want to show that the inequality also holds. Let us first define a variable for Wigner 3j symbols $\alpha_{\hat{l}} = (2\hat{l} + 1) \begin{pmatrix}
   l_1 & l_2 & \hat{l} \\
   -n & n & 0 
  \end{pmatrix} \begin{pmatrix}
   l_1 & l_2 & \hat{l} \\
   -k & k & 0 
  \end{pmatrix}$ and $c_{\hat{l}} = \sum_{p=1}^m P_{\hat{l}}(\cos \theta_p) \geq 0$ from Lemma \ref{ch5:lemm increasing}.
From Lemma \ref{ch5:lemm Wigner_3j_odd_even}, we have  $\sum_{\substack{\hat{l}=|l_2-l_1|\\\hat l,\text{even}}}^{l_2+l_1}\alpha_{\hat l}=0$ for  $-\text{min}\left(l_1,l_2\right) \leq k \neq n \leq \text{min}\left(l_1,l_2\right)$ and we can define an independent variable $\kappa =   \frac{c_{\card{l_1 -l_2}} + c_{l_1 + l_2}}{2}$. The upper bound of the product $\card{\sum_{p=1}^m \Wd{l_1}{k}{n}(\cos \theta_p) \Wd{l_2}{k}{n}(\cos \theta_p)}$ is given by
\begin{equation}\label{ch5:eq ub_kneqn}
\begin{aligned}
 \card{\sum_{\hat{l}=\card{l_1-l_2} \atop \hat{l},\text{even}}^{l_1 + l_2} \left(c_{\hat{l}} - \kappa\right)\alpha_{\hat{l}}} &\leq \frac{c_{\card{l_1 + l_2}} - c_{\card{l_1 -l_2}}}{2} \sum_{\hat{l} =\card{l_1-l_2} \atop \hat{l},\text{even}}^{l_1 + l_2} \card{\alpha_{\hat{l}}}\\
  &\leq \frac{\sqrt{2}}{4}\left(c_{2B-4} - c_{2}\right).
\end{aligned}
\end{equation}
The first inequality is derived by using the triangle inequality and the increasing property of the sum of equispaced samples of Legendre polynomials in Lemma \ref{ch5:lemm increasing}, i.e., $c_2 \leq c_4 \leq c_6 \leq \hdots \leq c_{2B-4}$. The last inequality also holds by the Cauchy-Schwarz inequality of
 $\alpha_{\hat{l}} = (2\hat{l} + 1) \begin{pmatrix}
   l_1 & l_2 & \hat{l} \\
   -n & n & 0 
  \end{pmatrix} \begin{pmatrix}
   l_1 & l_2 & \hat{l} \\
   -k & k & 0 
  \end{pmatrix}$.
 
 To be precise, for $-\text{min}\left(l_1,l_2\right) \leq k \neq n \leq\text{min}\left(l_1,l_2\right)$ one can write
\begin{equation*}
\begin{aligned}
\sum_{\hat{l}=\card{l_1-l_2} \atop \hat{l},\text{even}}^{l_1 + l_2}\card{(2\hat{l} + 1) \begin{pmatrix}
   l_1 & l_2 & \hat{l} \\
   -n & n & 0 
  \end{pmatrix} \begin{pmatrix}
   l_1 & l_2 & \hat{l} \\
   -k & k & 0 
  \end{pmatrix}}\leq \frac{1}{\sqrt{2}},
\end{aligned} 
\end{equation*}
where $\sum_{\hat{l}=\card{l_1-l_2} \atop \hat{l},\text{even}}^{l_1 + l_2} (2\hat{l} + 1) \begin{pmatrix}
   l_1 & l_2 & \hat{l} \\
   0 & 0 & 0 
  \end{pmatrix}^2= 1$ and $\sum_{\hat{l}=\card{l_1-l_2} \atop \hat{l},\text{even}}^{l_1 + l_2}(2\hat{l} + 1) \begin{pmatrix}
   l_1 & l_2 & \hat{l} \\
   -k & k & 0 
  \end{pmatrix}^2= \frac{1}{2}$ as a consequence of Lemma \ref{ch5:lemm Wigner_3j_odd_even}.
  
For $\tau = k = n \neq 0$ and $-\text{min}\left(l_1,l_2\right) \leq \tau  \leq\text{min}\left(l_1,l_2\right)$, let us write $ \beta_{\hat{l}} =(2\hat{l} + 1) 
  \begin{pmatrix}
   l_1&l_2& \hat{l} \\
   -\tau & \tau & 0 
  \end{pmatrix}^2$. The upper bound for product $\card{\sum_{p=1}^m \Wd{l_1}{\tau}{\tau}(\cos \theta_p) \Wd{l_2}{\tau}{\tau}(\cos \theta_p)}$ is given as
\begin{equation}\label{ch5:eq ub_keqneqtau}
\begin{aligned}
  \card{\sum_{\hat{l}=\card{l_1-l_2} \atop \hat{l},\text{even}}^{l_1 + l_2} c_{\hat{l}}\beta_{\hat{l}}} &\leq c_{2B-4}\sum_{\hat{l}=\card{l_1-l_2} \atop \hat{l},\text{even}}^{l_1 + l_2}\beta_{\hat{l}} = \frac{c_{2B-4}}{2}.
\end{aligned}
\end{equation}
The maximum of $c_{\hat{l}}$ is derived from the increasing property of the sum of equispaced samples of Legendre polynomials in Lemma \ref{ch5:lemm increasing}, i.e., $0 \leq c_2 \leq c_4 \leq c_6 \leq \hdots \leq c_{2B-4}$, and for $1\leq \card{k}=\card{n}=\card{\tau}\leq \text{min}(l_1,l_2)$, the sum of Wigner 3j symbols $\sum_{\hat{l}=\card{l_1-l_2} \atop \hat{l},\text{even}}^{l_1 + l_2}\beta_{\hat{l}} = \sum_{\hat{l}=\card{l_1-l_2} \atop \hat{l},\text{even}}^{l_1 + l_2}(2\hat{l} + 1) \begin{pmatrix}
   l_1 & l_2 & \hat{l} \\
   -\tau & \tau & 0 
  \end{pmatrix}^2 = \frac{1}{2}$ as the result from Lemma \ref{ch5:lemm Wigner_3j_odd_even}.
Therefore, from \eqref{ch5:eq ub_kneqn} and \eqref{ch5:eq ub_keqneqtau}, it is enough to consider the upper bound in \eqref{ch5:eq ub_keqneqtau}. We then need to show
\begin{equation}
\card{\sum_{p=1}^{m} P_{B-1}(\cos \theta_p)  P_{B-3}(\cos \theta_p)}- \frac{1}{2}\sum_{p=1}^{m} {P}_{2B-4}(\cos \theta_p)   \geq 0.
\end{equation}
We have the product of two Legendre polynomials $\card{\sum_{p=1}^m P_{B-3}(\cos \theta_p) P_{B-1}(\cos \theta_p)}$ as
$$\card{\sum_{\substack{\hat l=2\\\hat{l}, \text{even}}}^{2B-4} (2\hat{l} + 1) 
  \begin{pmatrix}
   B-3&B-1& \hat{l} \\
   0 & 0 & 0 
  \end{pmatrix}^2\sum_{p=1}^m P_{\hat{l}} (\cos \theta_p)}.$$
Suppose we have $\rho_{\hat{l}} = (2\hat{l} + 1) \begin{pmatrix}
   B-3 & B-1& \hat{l} \\
   0 & 0 & 0  
  \end{pmatrix}^2$ then we have to show $\sum_{\substack{\hat{l}=2\\\hat l,\text{even}}}^{2B-4} \rho_{\hat{l}}c_{\hat{l}}  -  \frac{c_{2B-4}}{2} \geq 0$. From Lemma \ref{ch5:lemm sum_equi}, for even $\hat{l}$ we have $c_{\hat l}=\sum_{p=1}^m P_{\hat{l}} (\cos \theta_p) = 1 + \frac{\hat l(\hat l+1)}{6(m-1)} + R_{\hat l}(m)$, where from Proposition \ref{ch5:prop residual} the interval of the residual is given as $-0.463 < R_{\hat{l}}(m) < 0 $. Finally, we can write 
\begin{equation*}
\begin{aligned}
\sum_{\substack{\hat{l}=2\\\hat l,\text{even}}}^{2B-4} \rho_{\hat{l}}c_{\hat{l}} &= \sum_{\substack{\hat{l}=2\\\hat l,\text{even}}}^{2B-4} \rho_{\hat{l}}\left(1 + \frac{\hat{l}(\hat{l}+1)}{6(m-1)} + R_{\hat{l}}(m)\right)
&\geq \sum_{\substack{\hat{l}=2\\\hat l,\text{even}}}^{2B-4} \rho_{\hat{l}}\left(0.537 + \frac{\hat{l}(\hat{l}+1)}{6(m-1)}\right).
\end{aligned}
\end{equation*}
Additionally, we can bound $\frac{c_{2B-4}}{2}$  as
\begin{equation*}
\begin{aligned}
\frac{c_{2B-4}}{2} &= \frac{1}{2}\left( 1+ \frac{(2B-4)(2B-3)}{6(m-1)} + R_{\hat{2B-4}}(m)\right)
&\leq \frac{1}{2}\left( 1+ \frac{(2B-4)(2B-3)}{6(m-1)}\right)
\end{aligned}
\end{equation*} 

Hence, the lower bound can be derived as $\sum_{\substack{\hat{l}=2\\\hat l,\text{even}}}^{2B-4} \rho_{\hat{l}}c_{\hat{l}}  -  \frac{c_{2B-4}}{2} \geq \hat{C} + \frac{B}{6(m-1)} > 0 $.
The constant $\hat{C}$ is obtained from $0.537 \sum_{\substack{\hat{l}=2\\\hat l\text{even}}}^{2B-4} \rho_{\hat{l}} - \frac{1}{2} = 0.537 - 0.5 > 0$, since $\sum_{\substack{\hat{l}=2\\\hat l,\text{even}}}^{2B-4} \rho_{\hat{l}} = 1$ as a consequence from \eqref{ch2:eq prop3j}. The last term $\frac{B}{6(m-1)}$ is deduced from Corollary \ref{ch5:cor max_B}, where $ \sum_{\hat{l}=2 \atop \hat{l}, \text{even}}^{2B -4} \rho_{\hat{l}}(\hat{l}^2 + \hat{l}) = 2 + 2(B-1)(B-2)$. Therefore, we have shown that the maximum is attained by the product of two Legendre polynomials for degrees $l_1 = B-3$ and $l_2 = B-1$.
\end{proof}


\section{Proofs of Lemmas and Proposition}
\label{sec:proof_of_lemmas}
\subsection{Proofs of Lemmas in Section \ref{sec3: sum_leg_express}}\label{proofsection3}
\begin{proof} [Proof of Lemma \ref{ch5:lemm sum_equi}] \label{prooflemma1}
The proof utilizes characterization of Legendre polynomials given in \eqref{App:eq Legendre_exp}. Note that the Legendre polynomial is an even function for even degrees. For $l = 0$, this summation is equal to $m$ regardless of how we sample the Legendre polynomials. Thus, the analysis is started for even degrees $l \geq 2$ and the samples are given by $\cos \theta_p = x_p=\frac{2p-m-1}{m-1}$ for $p \in [m]$. Since the Legendre polynomials are even and these sample points are symmetric on the interval $[-1,1]$, it is enough to only consider the positive samples. We first assume that $m$ is odd with $\widetilde{m}=(m-1)/2$,
\begin{equation}
\begin{aligned}
\sum_{p=1}^{{m}}P_l(x_p)  &= 2\sum_{p=1}^{\widetilde{m}}P_l(y_p)+ P_l(0),
\end{aligned}
\end{equation} 
where the samples in $(0,1]$ are given by $y_p = \frac{p}{\widetilde{m}}$ for $p \in  [\widetilde{m}]$. By using definition of Legendre polynomials in \eqref{App:eq Legendre_exp} and Bernoulli summation in \eqref{App:eq Faul_sum} we can write 
\begin{align*}
&\sum_{p=1}^{\widetilde{m}}P_l(y_p) = 2^l \sum_{k=0}^{l} \binom{l}{k} \binom{\frac{l+k-1}{2}}{l} \left( \frac{1}{k+1}\sum_{j=0}^k{B_j}\binom{k+1}{j} \widetilde{m}^{1-j}\right).
\end{align*}
Let expand the inner summation. For $k=0$ and from Bernoulli number $B_0 = 1$, Table \eqref{App:tab bernoulli_table}, we will have  
\begin{equation*}
\small
\begin{aligned}
\small
2^l \sum_{k=0}^{l} \binom{l}{k} \binom{\frac{l+k-1}{2}}{l} \frac{\widetilde{m}}{k+1} &=  2^l \sum_{k=0}^{l} \binom{l}{k} \binom{\frac{l+k-1}{2}}{l} \int_{0}^{1} \widetilde{m}x^{k} dx\\
&=\int_{0}^{1}P_l(x) \mrm dx = \frac 12 \int_{-1}^{1}P_l(x) \mrm dx=0
\end{aligned}
\end{equation*}
The last equality holds from the definition of Legendre polynomials in \eqref{App:eq Legendre_exp} and integration of Legendre polynomials on the interval $[-1,1]$, which is equal to $0$ for $l\neq 0$ as in \eqref{App:eq int_Leg}. It should be noted that, the Legendre polynomials are even functions for even degrees $l$.

For $k=1$ and use Bernoulli number $B_1 = \frac{1}{2}$, we have the following identity
\begin{equation}
\begin{aligned}
2^{l-1} \sum_{k=1}^{l} \binom{l}{k} \binom{\frac{l+k-1}{2}}{l}=\frac{P_l(1) - P_l(0)}{2},
\end{aligned}
\end{equation} 
where the equality is derived by using expansion of Legendre polynomials in \eqref{App:eq Legendre_exp} and substitute $x=1$. For even degree $l$ and $k=0$, we have $P_l(1) = 1$ and $P_l(0)$, respectively. Identity of $P_l(0)$ is derived in \eqref{App:eq even_deg_zeroLeg}. Hence, the second sum is equal to $\frac{1 - P_l(0)}{2}$. 

For $k=2$, the summation is then obtained by using Bernoulli number $B_2=\frac 16$ and the derivative of Legendre polynomials in \eqref{App:eq deriv_Leg}
\begin{equation*}
\small
\begin{aligned}
\sum_{k=2}^{l}2^l \binom{l}{k} \binom{\frac{l+k-1}{2}}{l} \frac{k}{12\widetilde{m}}  =\frac{1}{12\widetilde{m}} \left.\frac{\d P_{l}(x)}{\d x} \right\vert_{x=1}  = \frac{1}{12\widetilde{m}} \frac{l(l+1)}{2}. 
\end{aligned}
\end{equation*}
The final sum can be obtained as
\begin{equation*}
\begin{aligned}
\sum_{i=1}^{m}P_l(x_i) =2\sum_{i=0}^{\widetilde{m}}P_l(y_i) + P_l(0)
=1+\frac{l(l+1)}{6(m-1)}+ R_l(m),
\end{aligned}
\end{equation*} 
where we use the fact that $\widetilde{m} = \frac{m-1}{2}$. The remainder term of the summation, $R_l(m)$, can be expressed as
\begin{equation*}
\small
\begin{aligned}
   R_l(m) = 2^{l+1} \sum_{k=3}^{l} \binom{l}{k} \binom{\frac{l+k-1}{2}}{l}\left({\frac{1}{k+1}\sum_{j=3}^k{B_j}\binom{k+1}{j} \left(\frac{m-1}{2}\right)^{1-j}}\right).
\end{aligned}
\end{equation*}
 For $j \geq 3$, the Bernoulli number is non-zero only for even values of $j$, as discussed in Table \eqref{App:tab bernoulli_table}. By changing the summation index, we have $R_l(m)$ as
\begin{align*}
&\sum_{j=4 \atop j,\text{even} }^{l} B_j \frac{2^j}{j! (m-1)^{j-1}} 
    2^{l}\sum_{k=j}^l
    \binom{l}{k} \binom{\frac{l+k-1}{2}}{l} \frac{k!}{(k+1-j)!}.
\end{align*}
From \eqref{App:eq deriv_Leg}, \eqref{App:eq derivative_Leg_exp} and the relation between Bernoulli number and zeta function in  \eqref{App:eq bernoulli_zeta} we have  
\begin{equation*}
\begin{aligned}
\small 
\sum_{j=4 \atop j,\text{even} }^{l}   \frac{\zeta(j) (l + j -1)! 4 }{(j-1)!(l-j+1)! (2\pi)^j} \frac{(-1)^{\frac{j}{2} +1}}{(m-1)^{j-1}}= \sum_{j=4 \atop j,\text{even} }^{l}\frac{(-1)^{\frac{j}{2} + 1}S_l^j}{(m-1)^{j-1}}.\\
\end{aligned}
\end{equation*}

Now consider $j=2$ for the last equation with the value of $ \zeta(2) =  \frac{ \pi^2}{6}$, then the equation is equal to
\begin{equation}
\frac{\zeta(2) (l + 1)! 4  }{ (l- 1)! (2\pi)^2} \frac{1}{(m-1)} = \frac{ \zeta(2)(l+1)l}{\pi^2 (m-1)} = \frac{l(l+1)}{6(m-1)},
\end{equation}
which completes the claim. The same approach can be used to derive the result for even $m$.
\end{proof}

\subsection{Proof of Proposition in Section \ref{sec3: sum_leg_express}}
\begin{proof}[Proof of Proposition  \ref{ch5:prop residual}]
From Lemma \ref{ch5:lemm sum_equi}, the sum of equispaced samples Legendre polynomials can be written as
\begin{equation*}
\begin{aligned}
\sum_{p=1}^m P_l(\cos\theta_p) = 1 + \frac{l(l+1)}{6(m-1)} + R_l(m) ,
\end{aligned}
\end{equation*}
where $R_l(m) = \sum_{j= 4\atop j,\text{even} }^{l}\frac{(-1)^{\frac{j}{2} + 1}S_l^j}{(m-1)^{j-1}}$ with $S_l^j = \frac{\zeta(j) (l+j-1)! 4 }{(j-1)! (l-j+1)! (2\pi)^j} > 0$ . We want to show that the sequence of residual $\frac{S_l^j}{(m-1)^{j-1}}$ is decreasing for an increasing even $j \geq 4$. In other words, we want to show $\frac{S_l^j}{(m-1)^{j-1}} \geq \frac{S_l^{j+2}}{(m-1)^{j+1}}$ and write the ratio as
\begin{equation}
\begin{aligned} \label{ratio_k2}
\frac{S_l^{j+2} }{S_l^j (m-1)^{2}} &= \frac{\zeta(j+2) (l-j+1)(l-j)(l+j+1)(l+j)}{\zeta(j) (j+1)(j) ((m-1)2\pi)^2}\\
&<\frac{(l^2-j^2)\left((l^2-j^2) + 2l + 1\right) }{(j+1)(j) ((m-1)2\pi)^2}, 
\end{aligned}
\end{equation}
where upper bound is derived from the fact that the zeta function is decreasing for an increasing even $j$, that is $\zeta(j) > \zeta(j+2)$.
In order to show the decreasing property, it should be enough by showing that the ratio above is upper bounded by $1$, which is accomplished by considering $m - 1 \geq \frac{(l+1)^2}{10}$ for $4 \leq j_{,\text{even}} \leq  l$ and $l \geq  4$. 

Now, we want to show the lower bound of $R_l(m)$. For an even $\frac{l}{2}$, we will have $R_l(m)$ as $-\frac{ S_l^4}{(m-1)^{3}}  + \left(\frac{ S_l^6}{(m-1)^{5}}- \frac{S_l^8}{(m-1)^{7}} \right)+ \hdots + \left(\frac{S_l^{l-2}}{(m-1)^{l-3}} -\frac{S_l^l}{(m-1)^{l-1}} \right) \geq -\frac{ S_l^4}{(m-1)^{3}}.
$
The lower bound holds because $\frac{S_l^l}{(m-1)^{l-1}} $ is decreasing. Therefore, the subtraction in the bracket is positive. Using the geometric-arithmetic mean inequality  $\frac{(l+k)!}{(l-k)!} \leq (l+1)^{2k}$ \cite[eq.15]{lohofer_inequalities_1998}, we have $-\frac{ S_l^4}{(m-1)^{3}}  = -\frac{\zeta(4)(l+3)!4.10^3}{3!(l-3)!(2\pi)^4 (l+1)^6}  >  -0.463.$

For completeness, the same approach can be derived for an odd $\frac{l}{2}$. The difference is, instead of having two terms for $l-2$ and $l$ at the end, we only have $l$, which is positive because  $(-1)^{\frac{l}{2} + 1}$ is positive for an odd $\frac{l}{2}$. Thus, it does not change the lower bound. 

It is trivial to show that $\frac{S_l^{j }}{(m-1)^{j-1}}=\frac{\zeta(j) (l+j-1)! 4 }{(j-1)! (l-j+1)! (2\pi)^j(m-1)^{j-1}} $ converges to zero for sufficiently large samples $m$ compared to the degree $l$. Hence, giving the upper bound of the residual. 
\end{proof}

\begin{proof}[Proof of Lemma \ref{ch5:lemm increasing}]
For $l=2$, it is proven in Lemma \ref{ch5:lemm sum_equi} that  $\sum_{p=1}^m P_l(\cos \theta_p) = 1 +  \frac{l(l+1)}{6(m-1)} > 0$. Therefore, we have 
$\sum_{p=1}^m P_4(\cos \theta_p) - \sum_{p=1}^m P_2(\cos \theta_p)  = \frac{7}{3(m-1)} + R_4(m) =\frac{7}{3(m-1)} -\frac{7}{3(m-1)^3} \geq 0$,
 where $R_l(m)=\sum_{j = 4 \atop j, \text{even}}^l \frac{(-1)^{\frac{j}{2} + 1}S_l^j}{(m-1)^{j-1}}$ and $S_l^j = \frac{\zeta(j) (l+j-1)! 4 }{(j-1)! (l-j+1)! (2\pi)^j}$. Thus, we only need to prove for even $l  \geq 4$. For this reason, the increasing property of the summations, $\sum_{p=1}^m P_2(\cos \theta_p) < \sum_{p=1}^m P_4(\cos \theta_p) < \hdots < \sum_{p=1}^m P_{l, \text{even}}(\cos \theta_p)$, directly implies a non-negative property of the sum of equispaced samples Legendre polynomials. Since we compare $l+2$ to $l$, then the number of sample $m-1 \geq \frac{(l+3)^2}{10}$ should be considered, which means that for even $l \geq 4$, it is enough to show $\sum_{p=1}^m P_{l+2}(\cos \theta_p) - \sum_{p=1}^m P_l(\cos \theta_p)  \geq 0$. By using the result from Lemma \ref{ch5:lemm sum_equi}, the condition is equal to $\frac{2l+3}{3} \geq (m-1)\bigg(R_l(m) - R_{l+2}(m) \bigg)$.
 
Let observe the residual $(m-1)\bigg(R_l(m) - R_{l+2}(m) \bigg)$ and write as 
\begin{equation}\label{residual}
\begin{aligned}\sum_{j=4 \atop j, \text{even}}^{l} \frac{(-1)^{\frac{j}{2}} \left(S_{l+2}^j - S_{l}^j\right
)}{(m-1)^{j-2}} - \frac{(-1)^{\frac{l+2}{2} + 1}S_{l+2}^{l+2}}{(m-1)^{l}}.
\end{aligned}
\end{equation}  
First of all, we show that $\frac{\left(S_{l+2}^j - S_{l}^j\right
)}{(m-1)^{j-2}} > 0$ for a fix $m$ since
\begin{equation}\label{ratio_l2}
    \frac{S_{l+2}^j}{S_{l}^j} = \frac{(l+j+1)(l+j)}{(l-j+3)(l-j+2)} > 1.
\end{equation}
Second of all, the sequence $\frac{\left(S_{l+2}^j - S_{l}^j\right
)}{(m-1)^{j-2}}$ is decreasing if we increase even $j$, or we have
$
\frac{\big(S_{l+2}^j - S_{l}^j\big
)}{(m-1)^{j-2}}  > \frac{\big(S_{l+2}^{j+2} - S_{l}^{j+2}\big
)}{(m-1)^{j}}.$
Therefore, the expression is equivalent to showing that the ratio $\frac{\left(S_{l+2}^{j+2} - S_{l}^{j+2}\right
)}{\left(S_{l+2}^j - S_{l}^j\right) (m-1)^{2}}< 1$ holds for $m-1 \geq \frac{(l+3)^2}{10}$. Let expand this ratio by using \eqref{ratio_l2} as 
\begin{equation}\label{ratio2}
\begin{aligned}
\frac{\left(S_{l+2}^{j+2} - S_{l}^{j+2}\right
)}{\left(S_{l+2}^j - S_{l}^j\right) (m-1)^{2}} = \frac{S_l^{j+2} \big( \frac{(l+j+3)(l+j+2)}{(l-j+1)(l-j)}- 1\big)}{S_l^j \big(\frac{(l+j+1)(l+j)}{(l-j+3)(l-j+2)}- 1\big) (m-1)^{2}}  
\end{aligned}
\end{equation}
From \eqref{ratio_k2}, we know the ratio $\frac{S_l^{j+2}}{S_l^j (m-1)^2}$. Hence, \eqref{ratio2} can be expressed as
\begin{equation}\label{ch5:eq ratio_dua}
\begin{aligned}
 \frac{\zeta(j+2) (l+j+1)(l+j)}{\zeta(j)   (j) ((m-1)2\pi)^2}\frac{ (l-j+3)(l-j+2) }{  (j-1)} < 1. 
\end{aligned}
\end{equation}
The upper bound holds from the fact that $\zeta(j) > \zeta(j+2)$ and samples $m-1 \geq \frac{(l+3)^2}{10}$.  Thereby, it proves that for increasing even $j$, $\frac{\big(S_{l+2}^j - S_{l}^j\big
)}{(m-1)^{j-2}} $ is decreasing. Summarizing the results in \eqref{ratio_l2} and \eqref{ch5:eq ratio_dua} we have $ S^j_{l+2} > S^j_{l}$ and $\left(S^j_{l + 2}  - S^j_l\right)(m-1)^2 > \left(S^{j+2}_{l + 2}  - S^{j+2}_l\right)$. 

Therefore, for even $\frac{l}{2}$, we write 
\begin{equation*}
\begin{aligned}
\sum_{j=4 \atop j,\text{even}}^{l} \frac{(-1)^{\frac{j}{2}} \left(S_{l+2}^j - S_{l}^j\right
)}{(m-1)^{j-2}} &= \frac{\left(S_{l+2}^4 - S_{l}^4\right
)}{(m-1)^{2}} -\left( \frac{\left(S_{l+2}^6 - S_{l}^6\right
)}{(m-1)^{4}} - \frac{\left(S_{l+2}^8 - S_{l}^8\right
)}{(m-1)^{6}} \right) -\hdots  \\
 &-\left( \frac{\left(S_{l+2}^{l-2} - S_{l}^{l-2}\right
)}{(m-1)^{l-4}} - \frac{\left(S_{l+2}^l - S_{l}^l\right
)}{(m-1)^{l-2}} \right) \leq \frac{\left(S_{l+2}^4 - S_{l}^4\right
)}{(m-1)^{2}}.
\end{aligned}
\end{equation*} 
The upper bound follows because the subtractions in the brackets are positive. It should be noted that the upper bound also holds for an odd  $\frac{l}{2}$, where instead of having two terms for $l-2$ and $l$, we only have $l$ which is negative. Finally, collecting all these results, we can bound \eqref{residual} with $\frac{\big(S_{l+2}^4 - S_{l}^4\big
)}{(m-1)^{2}} + \frac{S_{l+2}^{l+2}}{(m-1)^{l}}$. We want to show this upper bound is smaller than $\frac{2l+3}{3}$.
    
Let first observe the upper bound of $ \frac{\big(S_{l+2}^4 - S_{l}^4\big
)}{(m-1)^{2}}$. For $m-1 \geq \frac{(l+3)^2}{10}$ and using the ratio in \eqref{ratio_l2}, we obtain 
\begin{equation*}
\begin{aligned}
\frac{\big(S_{l+2}^4 - S_{l}^4\big
)}{(m-1)^{2}} =\frac{S_{l}^4}{(m-1)^2}  \bigg( \frac{12l+18}{(l-1)(l-2)} \bigg)   &< \frac{\zeta(4)(24)10^2}{3!(2\pi)^4} \frac{(l+2)(l+1)l (2l+3)}{(l+3)^3} 
\\
& < C_1 (2l+3),
\end{aligned}
\end{equation*} 
where the constant $C_1 = 0.2778$ is derived from the fact that $\zeta(4) = \frac{\pi^4}{90}$ as in Table \eqref{App:tab zeta_table}. From the definition of $S_{l}^{j}$, the upper bound of $ \frac{S_{l+2}^{l+2}}{(m-1)^{l}}$ can be determined by $$\frac{S_{l+2}^{l+2}}{(m-1)^{l}}
\leq  \frac{\zeta(l+2)}{\pi^{2}}\frac{\Big(l+2 + \big(l+1\big)\Big)! }{\Big(l+2 - \big(l+1\big)\Big)!}\left(\frac{5}{\pi}\right)^l \frac{1}{(l+1)!(l+3)^{2l}}   < C_2 (2l+3)$$.
 
The inequality is derived from the geometric-arithmetic mean inequality for factorial $\frac{(l+k)!}{(l-k)!} \leq (l+1)^{2k}$ \cite[eq.15]{lohofer_inequalities_1998}. Additionally, by considering even degrees $l \geq 4$ we have $(l+3)^2 < (2l+3) (l+1)$ and $2^l < l!$. It is obvious that from decreasing property and the value of zeta function in Table \ref{App:tab zeta_table}, the maximum of the constant is achieved for $l = 4$, which gives $C_2 = 0.0422$.  

Combining the results we complete the proof $(m-1)\bigg(R_l(m) - R_{l+2}(m) \bigg) \leq (C_1 + C_2) (2l+3) < \frac{(2l+3)}{3}$. Thus, $\sum_{p=1}^m P_{l+2}(\cos \theta_p) - \sum_{p=1}^m P_l(\cos \theta_p)  \geq 0$.
\end{proof}
\subsection{Proofs of Lemmas in Section \ref{sec4: Ineq_3j_symbols}} \label{proofsection4}
\begin{proof}[Proof of Lemma \ref{ch5:lemm decreasingWigner3j}]
In \eqref{App:eq special_wigner3j}, we have the exact expression for Wigner 3j symbols $\begin{pmatrix}
   l_1 & l_2 & l_3 \\
   0 & 0 & 0
  \end{pmatrix}$, where $2L = l_1 + l_2 + l_3$. 
  
  The ratio between $\begin{pmatrix}
   l_1 & l_2 & l_3 \\
   0 & 0 & 0
\end{pmatrix}^2$ and $\begin{pmatrix}
   l_1+1 & l_2+1 & l_3 \\
   0 & 0 & 0
\end{pmatrix}^2$ can be written as  
\begin{equation}\small \label{ch5: eq wigner_decreasing_ratio}
\begin{aligned} 
&\frac{(2L+3)!(2L-2l_3)!}{(2L+1)!(2L+2-2l_3)!}\biggl(\frac{L!(L+1-l_3)!}{(L+1)!(L-l_3)!}\biggr)^2
&= \frac{(2L+3)(2L-2l_3+2)}{(2L+2)(2L-2l_3+1)}\geq 1.
\end{aligned}
\end{equation}
Therefore, it proves the first property $
\begin{pmatrix}
   l_1 & l_2 & l_3 \\
   0 & 0 & 0
\end{pmatrix}^2
\geq 
\begin{pmatrix}
   l_1+1 & l_2+1 & l_3 \\
   0 & 0 & 0
\end{pmatrix}^2
$. Similarly, for the second condition we can write the ratio between $\begin{pmatrix}
   l_1 & l_2 & l_3 \\
   0 & 0 & 0
\end{pmatrix}^2$ and $ 
\begin{pmatrix}
   l_1+2 & l_2  & l_3 \\
   0 & 0 & 0
\end{pmatrix}^2$ as $\frac{(L+\frac{3}{2})(L-l_3+1)}{(L+1)(L-l_3 +\frac{1}{2}) }\frac{(L-l_2+1)(L-l_1-\frac{1}{2})}{(L-l_2+\frac{1}{2})(L-l_1)}$. The last ratio can be written as $\frac{ L^2-Ll_1-Ll_2 +l_1l_2 + \frac{L}{2} +\frac{l_2}{2} - l_1 -\frac{1}{2}}{L^2-Ll_1 -Ll_2+l_1l_2+\frac{L}{2} - \frac{l_1}{2}}$.
To have the above ratio be greater than one, the following condition should be satisfied $
l_2 \geq l_1 +1  > l_1$.

This condition does not change the assumption in this chapter since we want to find the maximum for $l_1 \neq l_2$.
\end{proof} 
\begin{proof}[Proof of Lemma \ref{ch5:lemm Wigner_3j_odd_even}]
Let rewrite the product of Wigner d-functions as in \eqref{product_Wigner_small} for an arbitrary sample and $-\text{min}(l_1,l_2) \leq k \neq n \leq \text{min}(l_1,l_2)$
\begin{equation}
\begin{aligned}
 \card{\mathrm{d}_{l_1}^{k ,n }(\cos \theta) \mathrm{d}_{l_2}^{k ,n }(\cos \theta)} = \card{\sum_{\hat{l}=\card{l_1 - l_2}}^{l_1 + l_2} (2\hat{l} + 1) \begin{pmatrix}
   l_1 & l_2 & \hat{l} \\
   -k & k & 0 
  \end{pmatrix}  \begin{pmatrix}
   l_1 & l_2 & \hat{l} \\
   -n & n & 0 
  \end{pmatrix}P_{\hat{l}} (\cos \theta)}.
\end{aligned}
\end{equation}
Suppose we have $\theta = 0$. From the definition of Wigner d-functions in \eqref{Wigner_d}, we have the weight $\sin^{\xi} \left(\frac{0}{2} \right) = 0$ and therefore $\mathrm{d}_{l}^{k ,n }(\cos 0) = 0$. On the contrary, the Legendre polynomials become $P_{\hat l}(1) = 1$ for even and odd degrees $\hat l$, as discussed in \eqref{App:eq symm_assLeg}. Therefore, we have
\begin{equation} \label{ch5:eq odd_even}
\begin{aligned}
&\sum_{\hat{l},\text{even}} (2\hat{l}+1)
\begin{pmatrix}
   l_1 & l_2 & \hat{l}\\
   -k & k & 0
  \end{pmatrix}\begin{pmatrix}
   l_1 & l_2 & \hat{l}\\
   -n & n & 0
  \end{pmatrix}+ \sum_{\hat{l},\text{odd}} (2\hat{l}+1)\begin{pmatrix}
   l_1 & l_2 & \hat{l}\\
   -k & k & 0
  \end{pmatrix} 
\begin{pmatrix}
   l_1 & l_2 & \hat{l}\\
   -n & n & 0
  \end{pmatrix}=0,
\end{aligned}
\end{equation}
which is obvious because of the orthogonality of Wigner 3j symbols as discussed in \eqref{ch2:eq prop3j}. 

In contrast, if we consider $\theta = \pi$, we have the weight $\cos^{\lambda} \left(\frac{\pi}{2} \right) = 0$ and hence $\mathrm{d}_{l}^{k ,n }(\cos \pi) = 0$. From symmetry of the Legendre polynomials we have $P_{\hat l}(-1) = 1$ for even degrees ${\hat l}$ and $P_{\hat l}(-1) = -1$ for odd degrees $\hat{l}$. Hence, we obtain
\begin{equation} \label{ch5:eq odd_even2}
\begin{aligned}
&\sum_{\hat{l}, \text{even}} (2\hat{l}+1)
\begin{pmatrix}
   l_1 & l_2 & \hat{l}\\
   -k & k & 0
  \end{pmatrix} 
\begin{pmatrix}
   l_1 & l_2 & \hat{l}\\
   -n & n & 0
  \end{pmatrix}  = &\sum_{\hat{l}, \text{odd}} (2\hat{l}+1)
\begin{pmatrix}
   l_1 & l_2 & \hat{l}\\
   -k & k & 0
  \end{pmatrix} \begin{pmatrix}
   l_1 & l_2 & \hat{l}\\
   -n & n & 0
  \end{pmatrix}. 
\end{aligned}
\end{equation}
Using \eqref{ch5:eq odd_even} and \eqref{ch5:eq odd_even2}, we complete the proof
\begin{equation} 
\begin{aligned}
&\sum_{\hat{l}, \text{even}} (2\hat{l}+1)
\begin{pmatrix}
   l_1 & l_2 & \hat{l}\\
   -k & k & 0
  \end{pmatrix}  \begin{pmatrix}
   l_1 & l_2 & \hat{l}\\
   -n & n & 0
  \end{pmatrix}= &\sum_{\hat{l}, \text{odd}} (2\hat{l}+1)
\begin{pmatrix}
   l_1 & l_2 & \hat{l}\\
   -k & k & 0
  \end{pmatrix} \begin{pmatrix}
   l_1 & l_2 & \hat{l}\\
   -n & n & 0
  \end{pmatrix} = 0.
\end{aligned}
\end{equation}

For the case $k=n=\tau $ and $1 \leq  \card \tau \leq \text{min}(l_1,l_2)$, we can express the product of Wigner d-functions as
\begin{equation}
\begin{aligned}
 \card{\mathrm{d}_{l_1}^{\tau ,\tau }(\cos \theta) \mathrm{d}_{l_2}^{\tau ,\tau }(\cos \theta)} = \card{\sum_{\hat{l}=\card{l_1 - l_2}}^{l_1 + l_2} (2\hat{l} + 1) \begin{pmatrix}
   l_1 & l_2 & \hat{l} \\
   -\tau & \tau & 0 
  \end{pmatrix}^2 P_{\hat{l}} (\cos \theta)}.
\end{aligned}
\end{equation}
As discussed earlier, if we choose $\theta = \pi$, then we have
\begin{equation*} 
\begin{aligned}
\sum_{\hat{l}, \text{even}} (2\hat{l}+1)
\begin{pmatrix}
   l_1 & l_2 & \hat{l}\\
   -\tau & \tau & 0
  \end{pmatrix}^2 = \sum_{\hat{l}, \text{odd}} (2\hat{l}+1)
\begin{pmatrix}
   l_1 & l_2 & \hat{l}\\
   -\tau& \tau & 0
  \end{pmatrix}^2.
\end{aligned}
\end{equation*}
However, we know that the sum of squared Wigner 3j symbols for all $\hat{l}$ is $1$, due to the orthogonal property in \eqref{ch2:eq prop3j}, i.e., $\sum_{\hat{l}=\card{l_1 - l_2}}^{l_1 +l_2} (2\hat{l}+1)
\begin{pmatrix}
   l_1 & l_2 & \hat{l}\\
   -\tau & \tau & 0
  \end{pmatrix}^2 = 1$. Thus, we have sum for all even $\hat l$ or odd $\hat l$ as
$\sum_{\hat{l}, \text{even}} (2\hat{l}+1)
\begin{pmatrix}
   l_1 & l_2 & \hat{l}\\
   -\tau & \tau & 0
  \end{pmatrix}^2 = \sum_{\hat{l}, \text{odd}} (2\hat{l}+1)
\begin{pmatrix}
   l_1 & l_2 & \hat{l}\\
   -\tau & \tau & 0
  \end{pmatrix}^2 = \frac{1}{2} $.
\end{proof}

\begin{proof}[Proof of Lemma \ref{ch5:lemm Wigner_3j_integer}]
Let first define $ b^{(l_1 )}_{\hat{l}} = (2\hat{l}+1)
\begin{pmatrix}
   l_1 & l_1+2 & \hat{l}\\
   0 & 0 & 0
  \end{pmatrix}^2$. For increasing index $l_1$, we prove this lemma by using inductions.
  
For $l_1 = 0$, the result is $2+2(l_1 + 2)(l_1 +1) = 6$. The summation can be written as $\sum_{\hat{l}=\card{l_1-l_2}=2 \atop \hat{l},\text{even}}^{2}  b^{(l_1 = 0)}_{\hat{l}}(\hat{l}^2 + \hat{l}) = b^{(l_1=0)}_2 6 = 6$, 
which is true because of the orthogonal property of Wigner 3j symbols, as discussed in \eqref{ch2:eq prop3j}, $\sum_{\hat{l}=\card{l_1-l_2}=2 \atop \hat{l},\text{even}}^{2l_1 + 2}  b_{\hat{l}} = 1$.

For $l_1 = 1$, we have $2+2(l_1 + 2)(l_1 +1) = 14$. This summation becomes complicated since we have two different values $\hat{l}$ for the Wigner 3j symbols, $$\sum_{\hat{l}=\card{l_1-l_2}=2 \atop \hat{l},\text{even}}^{4}  b^{(l_1 = 1)}_{\hat{l}}(\hat{l}^2 + \hat{l}) = b^{(l_1=1)}_2 6 + b^{(l_1=1)}_4 20$$.
Since we have the ratio between two consecutive Wigner 3j symbols for fixed values of $\hat{l}$ as discussed in the proof of Lemma \ref{ch5:lemm decreasingWigner3j} in \eqref{ch5: eq wigner_decreasing_ratio}, the relation between two different values $l_1$ for Wigner 3j symbols, $b^{(l_1+1)}_{\hat{l}}=\begin{pmatrix}
   l_1 + 1 & l_1 + 3 &  \hat{l} \\
   0 & 0 & 0
\end{pmatrix}^2$ and $b^{(l_1)}_{\hat{l}} = \begin{pmatrix}
   l_1  & l_1+2 & \hat{l} \\
   0 & 0 & 0
\end{pmatrix}^2$ can be obtained as 
\begin{equation}\label{ratio}
\frac{b^{(l_1+1)}_{\hat{l}}}{b^{(l_1)}_{\hat{l}}}=\frac{(2l_1+4 + \hat{l})(2l_1 + 3 - \hat{l})}{(2l_1+5 + \hat{l})(2l_1+4 - \hat{l})} = C_{\hat{l}}^{(l_1)} 
\end{equation}

From this relation, we can write  $\sum_{\hat{l}=\card{l_1-l_2}=2 \atop \hat{l}, \text{even}}^{4}  b^{(l_1 = 1)}_{\hat{l}}(\hat{l}^2 + \hat{l}) = C_{2}^{(l_1=0)}b^{(l_1=0)}_2 6 + b^{(l_1=1)}_4 20 = \frac{3}{7} 6 + \bigg(1 -\frac{3}{7}\bigg) 20 = 14$,
which is correct for $l_1=1$. The equality is derived from the fact that $C_{2}^{(l_1=0)} = \frac{3}{7}$ and from the previous case, $l_1 = 0$, we have $b^{(l_1=0)}_2 = 1$. Additionally, from \eqref{ch2:eq prop3j}, the summation is  $b_2^{(l_1=1)} + b_4^{(l_1=1)} = 1$.

Let generalize the induction part and consider the assumption for $l_1 = k$
\begin{equation}\label{ch5:eq induction_part}
\begin{aligned}
\sum_{\hat{l}= 2 \atop \hat{l}, \text{even}}^{2k + 2}  b^{(l_1 = k)}_{\hat{l}}(\hat{l}^2 + \hat{l}) = 2 + 2(k+2)(k +1).
\end{aligned}
\end{equation}
Therefore, we can determine the induction part to observe, where we have $$\sum_{\hat{l}= 2 \atop \hat{l}, \text{even}}^{2k + 2}  b^{(l_1 + 1 =k+1)}_{\hat{l}}(\hat{l}^2 + \hat{l}) + b^{(l_1 + 1 =k+1)}_{2k+4}\left((2k+4)^2 + (2k +4)\right)$$ 
From \eqref{ratio} we can write above summation as
$$\sum_{\hat{l}= 2 \atop \hat{l}, \text{even}}^{2k + 2}  C^{(l_1=k)}_{\hat{l}}b^{(l_1=k)}_{\hat{l}}(\hat{l}^2 + \hat{l}) + \left(1  - \sum_{\hat{l}= 2 \atop \hat{l}, \text{even}}^{2k + 2} C^{(l_1=k)}_{\hat{l}}b^{(l_1=k)}_{\hat{l}}\right)\left((2k+4)^2 + (2k +4)\right),$$ where we have $C^{(l_1=k)}_{\hat{l}} = \frac{4k^2 + 14k + 12 - (\hat{l}^2 + \hat{l})}{4k^2 + 18k + 20-(\hat{l}^2 + \hat{l})}$. Thus, we obtain $$C^{(l_1=k)}_{\hat{l}}\left(\hat{l}^2 + \hat{l} - \left((2k+4)^2 + (2k +4)\right)\right) = - \left(4k^2 + 14k + 12 \right) + (\hat{l}^2 + \hat{l}).$$ 

From \eqref{ch5:eq induction_part} we have $\sum_{\hat{l}= 2 \atop \hat{l}, \text{even}}^{2k + 2}  b^{(l_1 = k)}_{\hat{l}}(\hat{l}^2 + \hat{l}) = 2 + 2(k+2)(k +1)$. Combining these results, we can write the summation as $\big(4k + 8\big) + 2 + 2(k+2)(k+1)= 2 + 2(k+3)(k+2)$ and complete the proof.
 
\end{proof}

\subsection{A remark on norms of the columns}

The focus of our derivations has been on the inner product of the columns without the normalization.  
In this section, we study more closely the $\ell_2$-norm of the columns and provide some indications of why these norms do not contribute to the main inequality.   
An approximation of the $\ell_2$-norm of equispaced samples Wigner d-functions is given in the following lemma. 
\begin{lemma}\label{ch5:lemm l2normWigner}
Suppose we have a vector of sampled Wigner d-functions $\Wd l {k}{n}(\cos \bs \theta)\defeq\left(\Wd l{k}{n}(\cos\theta_1),\dots,\Wd l{k}{n}(\cos\theta_m)\right)^T$ with sampling points as in \eqref{ch5:eq equispaced}. The $\ell_2$-norm of this vector can be approximated by 
\begin{equation*}
\norm{\Wd l {k}{n}(\cos \bs \theta)}_2^2 = \frac{m-1}{2l+1} + D_1(k,n) +  \mathcal{O}(m^{-1}),
\end{equation*}
where  
\begin{equation}
\small
D_1(k,n)=\frac{\card{\mathrm{d}^{k,n}_l(1)}^2 + \card{\mathrm{d}^{k,n}_l(-1)}^2}{2} =
\begin{cases} 
      \frac{1}{2} & \text{for $ {k}= {n} \neq 0$} \\
      1 & \text{for $k = n =0$} \\
      0 & \text{for $ {k} \neq  {n}$}.  
   \end{cases}
\end{equation}
\end{lemma}

\begin{proof}[Proof of Lemma \ref{ch5:lemm l2normWigner}]
Wigner d-functions are continuous and integrable on the interval $[-1,1]$. If we have sampling points $x_p = \cos \theta_p = \frac{2(p-1)}{m-1}-1$ for $p \in [m]$ with a vector of equispaced points $\mbf x = \cos \bs \theta \in \R^m$, then we have distance between the sampling points $\Delta_x = \frac{2}{m-1}$. Therefore, we can write this summation as a trapezoidal rule of the Riemannian sum.
\begin{equation*}
\small
\begin{aligned}
&\Delta_x\Bigg(\frac{\big|\mathrm{d}^{k,n}_l(1)\big|^2 + \big|\mathrm{d}^{k,n}_l(-1)\big|^2}{2}  + \sum_{p=2}^{m-1} \big|\mathrm{d}^{k,n}_l\big((p-1)\Delta_x -1\big)\big|^2\Bigg)\\&
= \Delta_x\Bigg(\norm{\mathrm{d}_l^{k,n}(\boldsymbol x)}_2^2 - D_1(k,n)\Bigg),\\ 
\end{aligned}
\end{equation*}
where $ D_1(k,n) = \frac{\big|\mathrm{d}^{k,n}_l(1)\big|^2 + \big|\mathrm{d}^{k,n}_l(-1)\big|^2}{2}$. It is well known that the integral of squared Wigner d-functions, as in \eqref{inner_prod_Wigner}, is given by
\begin{equation*}
\begin{aligned}
\int_{-1}^1 \big|\mathrm{d}^{k,n}_l(x)\big|^2  dx= \frac{2}{2l+1}.
\end{aligned}
\end{equation*}
Thereby, the approximation error can be written as  
\begin{equation*}
\begin{aligned}
&\left|\int_{-1}^1 \big|\mathrm{d}^{k,n}_l(x)\big|^2 dx - \Delta_x\left(\norm{\mathrm{d}_l^{k,n}(\boldsymbol x)}_2^2 - D_1(k,n)\right) \right|  \\&=\left|\frac{m-1}{2l+1} + D_1(k,n) - \norm{\mathrm{d}_l^{k,n}(\boldsymbol x)}_2^2\right| = \mcl{O}(m^{-1}), \\
\end{aligned}
\end{equation*}
where the $\mcl{O}(m^{-1})$ is well-known error approximation from the Riemannian sum since we have $\Delta_x = \frac{2}{m-1}$. Expressed differently, we can write the summation formula as 
\begin{equation*}
\norm{\Wd l {k}{n}(\cos \bs \theta)}_2^2 = \frac{m-1}{2l+1} + D_1(k,n) +  \mathcal{O}(m^{-1}).
\end{equation*}
It is important to have a closed-form expression of $D_1(k,n)$. From \eqref{Wigner_d}, we know that the Wigner d-functions are weighted Jacobi polynomials. For several conditions of $-l \leq k,n \leq l$, we can get different $\xi=\card{k-n}$, $\lambda=\card{k+n}$, and $\alpha=l-\big(\frac{\xi+\lambda}{2}\big)$ on Jacobi polynomials, which change the value of constant $D_1(k,n)$. Those conditions are given in the following:
\begin{itemize}
\item For $k=n=0$, we have $\lambda=\xi=0$ and the Wigner d-function becomes Legendre polynomial $\mathrm{d}_l^{0,0}(\cos \theta) = P_l(\cos \theta)$. Hence, $\big|\mathrm{d}_l^{0,0}(1)\big|^2 = \big|\mathrm{d}_l^{0,0}(-1)\big|^2 = 1$ because of the symmetry of Legendre polynomials in \eqref{App:eq symm_assLeg}. Therefore, we obtain 
\begin{equation*}
\norm{P_l(\cos \boldsymbol \theta)}_2^2 = 1 + \frac{m-1}{2l+1} + \mathcal{O}(m^{-1}).
\end{equation*}


\item For $ {k} \neq {n}$, we have $\big|\mathrm{d}_l^{k,n}(1)\big|^2 = \big|\mathrm{d}_l^{k,n}(-1)\big|^2 = D_1(k,n)=0$. This is because for $\theta = 0$ or $ \theta = \pi$, the weight of the Wigner d-functions are $\sin^{\xi} \bigg(\frac{\theta}{2}\bigg) = 0$ or $\cos^{\lambda}\bigg(\frac{\theta}{2}\bigg) = 0$. Hence, we have 
\begin{equation*}
\norm{\Wd{l}{k}{n}(\cos \bs \theta)}_2^2 = \frac{m-1}{2l+1} + \mathcal{O}(m^{-1}).
\end{equation*}
For a specific case $k=0$ or $n=0$, then $\xi=\lambda$ and the Wigner d-functions become associated Legendre polynomials $\mathrm{d}_l^{k,0}(\cos \theta) = C_l^k P_l^k(\cos \theta)$, where $C_l^k=\sqrt{\frac{(l-k)!}{(l+k)!}}$, as given in \eqref{legend_jacobi}. Therefore, the $\ell_2$-norm of associated Legendre polynomials is 
\begin{equation*}
\norm{C_{l}^kP^k_l(\cos \boldsymbol \theta)}_2^2 = \frac{m-1}{2l+1} + \mathcal{O}(m^{-1}).
\end{equation*}
\item If $ {k}= {n} = {\tau} \neq 0$, then $\xi=0$ and $\lambda = 2\card{\tau}$ or $\xi=2\card{\tau}$ and $\lambda = 0$. The Wigner d-functions are $\big|\mathrm{d}_l^{\tau,\tau}(1)\big|^2 = \big|P_{\alpha}^{0,\lambda}(1)\big|^2 = \binom{\alpha}{\alpha} = 1$ and $\big|\mathrm{d}_l^{\tau,\tau}(-1)\big|^2=0$ or vice versa, because the weight of the Wigner d-functions are $\cos^{\lambda} \bigg(\frac{\pi}{2}\bigg) = 0$ or $\sin^{\xi} \bigg(\frac{0}{2}\bigg) = 0$ and due to the property of the Jacobi polynomials in \eqref{App:eq symm_Jacobi}. Thus, we have $D_1(k,n)=\frac{1}{2}$. 
\end{itemize}
From those characterizations of $D_1(k,n)$, we complete the proof.
\end{proof}
The proof relies heavily on the definition of the Wigner d-functions in \eqref{Wigner_d} and Jacobi polynomials. 

There are some immediate corollaries from this Lemma. First, for equal orders $k=n$, the norm is decreasing in the degree $l$, which means that the ordering between inner products is preserved after the division. 
This norm does not have a strong ordering between different degrees and orders (for instance from $k=n=0$ to $k=n\neq 0$). However, 
the $\ell_2$-norm of Wigner d-functions and associated Legendre polynomials are approximately the same for a sufficiently large $m$. Note that for large enough $m$, the norm, after division by $m$, approaches the functional $L_2$-norm of Wigner d-functions given by $2/(2l+1)$.


\bibliographystyle{IEEEtran}
\bibliography{reference}

\begin{thebibliography}{10}
\providecommand{\url}[1]{#1}
\csname url@samestyle\endcsname
\providecommand{\newblock}{\relax}
\providecommand{\bibinfo}[2]{#2}
\providecommand{\BIBentrySTDinterwordspacing}{\spaceskip=0pt\relax}
\providecommand{\BIBentryALTinterwordstretchfactor}{4}
\providecommand{\BIBentryALTinterwordspacing}{\spaceskip=\fontdimen2\font plus
\BIBentryALTinterwordstretchfactor\fontdimen3\font minus
  \fontdimen4\font\relax}
\providecommand{\BIBforeignlanguage}[2]{{%
\expandafter\ifx\csname l@#1\endcsname\relax
\typeout{** WARNING: IEEEtran.bst: No hyphenation pattern has been}%
\typeout{** loaded for the language `#1'. Using the pattern for}%
\typeout{** the default language instead.}%
\else
\language=\csname l@#1\endcsname
\fi
#2}}
\providecommand{\BIBdecl}{\relax}
\BIBdecl

\bibitem{bangun_sensing_2020}
A.~Bangun, A.~Behboodi, and R.~Mathar, ``Sensing {Matrix} {Design} and {Sparse}
  {Recovery} on the {Sphere} and the {Rotation} {Group},'' \emph{IEEE
  Transactions on Signal Processing}, pp. 1--1, 2020.

\bibitem{bangun_coherence_2018}
------, ``Coherence {Bounds} for {Sensing} {Matrices} in {Spherical}
  {Harmonics} {Expansion},'' in \emph{2018 {IEEE} {International} {Conference}
  on {Acoustics}, {Speech} and {Signal} {Processing} ({ICASSP})}.\hskip 1em
  plus 0.5em minus 0.4em\relax Calgary, AB: IEEE, Apr. 2018, pp. 4634--4638.

\bibitem{culotta-lopez_compressed_2018}
C.~Culotta-Lopez, D.~Heberling, A.~Bangun, A.~Behboodi, and R.~Mathar, ``A
  {Compressed} {Sampling} for {Spherical} {Near}-{Field} {Measurements},'' in
  \emph{2018 {AMTA} {Proceedings}}, Nov. 2018, pp. 1--6, iSSN: 2474-2740.

\bibitem{burq_weighted_2012}
N.~Burq, S.~Dyatlov, R.~Ward, and M.~Zworski,
  ``\BIBforeignlanguage{en}{Weighted {Eigenfunction} {Estimates} with
  {Applications} to {Compressed} {Sensing}},''
  \emph{\BIBforeignlanguage{en}{SIAM Journal on Mathematical Analysis}},
  vol.~44, no.~5, pp. 3481--3501, Jan. 2012.

\bibitem{rauhut_sparse_2011}
H.~Rauhut and R.~Ward, ``Sparse recovery for spherical harmonic expansions,''
  \emph{Proceedings of 9th International Conference on Sampling Theory and
  Applications (SampTA 2011)}, Feb. 2011, arXiv: 1102.4097.

\bibitem{tillmann2014computational}
A.~M. Tillmann and M.~E. Pfetsch, ``The computational complexity of the
  restricted isometry property, the nullspace property, and related concepts in
  compressed sensing,'' \emph{Information Theory, IEEE Transactions on},
  vol.~60, no.~2, pp. 1248--1259, 2014.

\bibitem{bandeira2013certifying}
A.~S. Bandeira, E.~Dobriban, D.~G. Mixon, and W.~F. Sawin, ``Certifying the
  restricted isometry property is hard,'' \emph{IEEE transactions on
  information theory}, vol.~59, no.~6, pp. 3448--3450, 2013.

\bibitem{dougall_product_1953}
J.~Dougall, ``\BIBforeignlanguage{en}{The {Product} of {Two} {Legendre}
  {Polynomials}},'' \emph{\BIBforeignlanguage{en}{Proceedings of the Glasgow
  Mathematical Association}}, vol.~1, no.~3, pp. 121--125, Sep. 1953.

\bibitem{adams_iii_1878}
C.~Adams, J, ``\BIBforeignlanguage{en}{On the expression of the product of any
  two {L}egendre's coefficients by means of a series of {L}egendre's
  coefficients},'' \emph{\BIBforeignlanguage{en}{Proceedings of the Royal
  Society of London}}, vol.~27, no. 185-189, pp. 63--71, Dec. 1878.

\bibitem{gasper_linearization_1970}
G.~Gasper, ``\BIBforeignlanguage{en}{Linearization of the {Product} of {Jacobi}
  {Polynomials}. {I}},'' \emph{\BIBforeignlanguage{en}{Canadian Journal of
  Mathematics}}, vol.~22, no.~1, pp. 171--175, Feb. 1970.

\bibitem{george_gasper_linearization_1970}
{George Gasper}, ``\BIBforeignlanguage{en}{Linearization of the {Product} of
  {Jacobi} {Polynomials}. {II}},'' \emph{\BIBforeignlanguage{en}{Canadian
  Journal of Mathematics}}, vol.~22, no.~3, pp. 582--593, Jun. 1970.

\bibitem{askey_linearization_1971}
R.~Askey and G.~Gasper, ``\BIBforeignlanguage{en}{Linearization of the
  {Product} of {Jacobi} {Polynomials}. {III}},''
  \emph{\BIBforeignlanguage{en}{Canadian Journal of Mathematics}}, vol.~23,
  no.~2, pp. 332--338, Apr. 1971.

\bibitem{edmonds_angular_1974}
A.~R. Edmonds, \emph{Angular momentum in quantum mechanics}, 3rd~ed., ser.
  Investigations in physics.\hskip 1em plus 0.5em minus 0.4em\relax Princeton,
  N.J: Princeton University Press, 1974, no.~4.

\bibitem{rose_elementary_1995}
M.~E. Rose, \emph{Elementary theory of angular momentum}.\hskip 1em plus 0.5em
  minus 0.4em\relax New York: Dover, 1995.

\bibitem{simons_spatiospectral_2006}
F.~J. Simons, F.~A. Dahlen, and M.~A. Wieczorek,
  ``\BIBforeignlanguage{en}{Spatiospectral {Concentration} on a {Sphere}},''
  \emph{\BIBforeignlanguage{en}{SIAM Review}}, vol.~48, no.~3, pp. 504--536,
  Jan. 2006.

\bibitem{kondor_clebsch_2018}
R.~Kondor, Z.~Lin, and S.~Trivedi, ``Clebsch-{G}ordan {N}ets : a {F}ully
  {F}ourier {S}pace {S}pherical {C}onvolutional {N}eural {N}etwork,'' in
  \emph{Advances in {Neural} {Information} {Processing} {Systems} 31},
  S.~Bengio, H.~Wallach, H.~Larochelle, K.~Grauman, N.~Cesa-Bianchi, and
  R.~Garnett, Eds.\hskip 1em plus 0.5em minus 0.4em\relax Curran Associates,
  Inc., 2018, pp. 10\,117--10\,126.

\bibitem{culotta-lopez_practical_2019}
C.~Culotta-Lopez, B.~Walkenhorst, Q.~Ton, and D.~Heberling, ``Practical
  {Considerations} in {Compressed} {Spherical} {Near}-{Field} {Measurements},''
  in \emph{2019 {Antenna} {Measurement} {Techniques} {Association} {Symposium}
  ({AMTA})}.\hskip 1em plus 0.5em minus 0.4em\relax San Diego, CA, USA: IEEE,
  Oct. 2019, pp. 1--6.

\bibitem{rauhut_sparse_2012}
H.~Rauhut and R.~Ward, ``\BIBforeignlanguage{en}{Sparse {Legendre} expansions
  via $\ell_1$-minimization},'' \emph{\BIBforeignlanguage{en}{Journal of
  Approximation Theory}}, vol. 164, no.~5, pp. 517--533, May 2012.

\bibitem{rauhut_interpolation_2016}
------, ``\BIBforeignlanguage{en}{Interpolation via {W}eighted
  $\ell_1$-minimization},'' \emph{\BIBforeignlanguage{en}{Applied and
  Computational Harmonic Analysis}}, vol.~40, no.~2, pp. 321--351, Mar. 2016.

\bibitem{hofmann_minimum_2019}
B.~Hofmann, O.~Neitz, and T.~F. Eibert, ``On the {Minimum} {Number} of
  {Samples} for {Sparse} {Recovery} in {Spherical} {Antenna} {Near}-{Field}
  {Measurements},'' \emph{IEEE Transactions on Antennas and Propagation},
  vol.~67, no.~12, pp. 7597--7610, Dec. 2019.

\bibitem{hampton_compressive_2015}
J.~Hampton and A.~Doostan, ``\BIBforeignlanguage{en}{Compressive {S}ampling of
  {P}olynomial {C}haos {E}xpansions: {Convergence} {A}nalysis and {S}ampling
  {S}trategies},'' \emph{\BIBforeignlanguage{en}{Journal of Computational
  Physics}}, vol. 280, pp. 363--386, Jan. 2015.

\bibitem{foucart2013mathematical}
S.~Foucart and H.~Rauhut, \emph{A mathematical introduction to compressive
  sensing}.\hskip 1em plus 0.5em minus 0.4em\relax Basel, Switzerland:
  Birkhäuser, 2013.

\bibitem{schulten_exact_1975}
K.~Schulten and R.~G. Gordon, ``\BIBforeignlanguage{en}{Exact {R}ecursive
  {E}valuation of 3j and 6j {C}oefficients for {Q}uantum {M}echanical
  {C}oupling of {A}ngular {M}omenta},'' \emph{\BIBforeignlanguage{en}{Journal
  of Mathematical Physics}}, vol.~16, no.~10, pp. 1961--1970, Oct. 1975.

\bibitem{zeiler_adadelta_2012}
M.~D. Zeiler, ``{ADADELTA}: {An} {Adaptive} {Learning} {Rate} {Method},''
  \emph{arXiv:1212.5701 [cs]}, Dec. 2012, arXiv: 1212.5701.

\bibitem{duchi_adaptive_2011}
J.~Duchi, E.~Hazan, and Y.~Singer, ``Adaptive {Subgradient} {Methods} for
  {Online} {Learning} and {Stochastic} {Optimization},'' \emph{Journal of
  Machine Learning Research}, vol.~12, no. Jul, pp. 2121--2159, 2011.

\bibitem{kingma_adam_2017}
D.~P. Kingma and J.~Ba, ``Adam: {A} {Method} for {Stochastic} {Optimization},''
  \emph{3rd International Conference for Learning Representations}, Jan. 2017,
  arXiv: 1412.6980.

\bibitem{welch_lower_1974}
L.~Welch, ``\BIBforeignlanguage{en}{Lower {B}ounds on the {M}aximum {C}ross
  {C}orrelation of {S}ignals},'' \emph{\BIBforeignlanguage{en}{IEEE
  Transactions on Information Theory}}, vol.~20, no.~3, pp. 397--399, May 1974.

\bibitem{lohofer_inequalities_1998}
G.~Loh{\"o}fer, ``\BIBforeignlanguage{en}{Inequalities for the {Associated}
  {Legendre} {Functions}},'' \emph{\BIBforeignlanguage{en}{Journal of
  Approximation Theory}}, vol.~95, no.~2, pp. 178--193, Nov. 1998.

\bibitem{bailey_generalized_1964}
W.~N. Bailey, \emph{\BIBforeignlanguage{English}{Generalized {Hypergeometric}
  {Series}.}}\hskip 1em plus 0.5em minus 0.4em\relax Stechert-Hafner Service
  Agency, 1964.

\bibitem{laurent_scaling_2017}
G.~M. Laurent and G.~R. Harrison, ``The {S}caling {P}roperties and the
  {M}ultiple {D}erivative of {Legendre} {P}olynomials,'' \emph{arXiv:1711.00925
  [math]}, Oct. 2017, arXiv: 1711.00925.

\bibitem{abramowitz_handbook_1965}
M.~Abramowitz and I.~A. Stegun, \emph{\BIBforeignlanguage{Englisch}{Handbook of
  {Mathematical} {Functions}}}.\hskip 1em plus 0.5em minus 0.4em\relax New
  York, NY: Dover Publications Inc., Jun. 1965.

\bibitem{knuth_johann_1993}
D.~E. Knuth, ``Johann {Faulhaber} and {Sums} of {Powers},'' \emph{Mathematics
  of Computation}, vol.~61, no. 203, pp. 277--294, 1993, publisher: American
  Mathematical Society.

\bibitem{jacobi_c_2013}
C.~G.~J. Jacobi, \emph{C. {G}. {J}. {Jacobi}'s {Gesammelte} {Werke}:
  {Herausgegeben} auf {Veranlassung} der k{\"o}niglich preussischen {Akademie}
  der {Wissenschaften}}, K.~Weierstrass, Ed.\hskip 1em plus 0.5em minus
  0.4em\relax Cambridge: Cambridge University Press, 2013.

\bibitem{abel_untersuchungen_1895}
N.~H. Abel, \emph{\BIBforeignlanguage{fre}{Untersuchungen {\"u}ber die {Reihe}:
  $1 + (m/1)x + m·(m - 1)/(1·2)·x^2+ m·(m - 1)·(m - 2)/(1·2·3)·x^3+
  ...$}}.\hskip 1em plus 0.5em minus 0.4em\relax Leipzig, W. Engelmann, 1895.

\bibitem{kennedy_hilbert_2013}
R.~A. Kennedy and P.~Sadeghi, \emph{Hilbert {S}pace {M}ethods in {S}ignal
  {P}rocessing}.\hskip 1em plus 0.5em minus 0.4em\relax Cambridge, UK:
  Cambridge University Press, 2013, oCLC: ocn835955494.

\end{thebibliography}
\newpage

\section{Supplementary Materials} \label{Supp_Material}
\subsection{Derivatives}
\subsubsection{Derivative of spherical harmonics}\label{Deriv}
In this article, we implement gradient descent based algorithms to optimize sampling points on the sphere. Hence, it is necessary to mention derivative of spherical harmonics with respect to $\theta$ and $\phi$ as follows.
\begin{equation*}
\small
\begin{aligned}
\frac{\partial \Y{l}{k} (\theta,\phi)}{\partial \theta} &= \frac{k \Y{l}{k} (\theta,\phi)}{\tan \theta}  + \sqrt{(l-k)(l+k+1)} \Y{l}{k+1}(\theta,\phi) e^{-i\phi}, \\
\frac{\partial \Y{l}{k} (\theta,\phi)}{\partial \phi} &= ik \Y{l}{k}(\theta,\phi),
\end{aligned}   
\end{equation*}
where the parameters are given in \eqref{sec2:eq gener_SH}.
Since we want to minimize the product of two spherical harmonics, the derivative rule of a product is applied.
\subsubsection{Derivative of Wigner D-functions}
Similar to the spherical harmonics case, the derivative of Wigner D-function with respect to the $\theta,\phi$ and $\chi$ are used, as stated in the following.
\begin{equation*}
\small
\begin{aligned}
\frac{\partial \D{l}{k}{n}(\theta, \phi,\chi)}{\partial \theta} &= \left( \frac{\lambda \sin^2 \theta}{2(1 + \cos \theta)} - \frac{\xi \sin^2 \theta}{2(1 - \cos \theta)}  \right) \D{l}{k}{n} (\theta,\phi, \chi)\\
&- \sin \theta\,\omega \sqrt{\gamma}\left(\frac{\xi + \lambda + \alpha + 1}{2}\right) \sin^{\xi}\left(\frac{\theta}{2} \right) \cos^{\lambda} \left( \frac{\theta}{2}\right) \\
&\times P_{\alpha-1}^{\xi+1,\lambda + 1}(\cos \theta) e^{-i(k\phi + n\chi)}\\
\frac{\partial \D{l}{k}{n}(\theta, \phi,\chi)}{\partial \phi} &= -ik \D{l}{k}{n}(\theta, \phi, \chi)\,,\\
\frac{\partial \D{l}{k}{n}(\theta, \phi,\chi)}{\partial \chi} &= -in\D{l}{k}{n}(\theta, \phi,\chi)
\end{aligned}
\end{equation*}
where the parameters are given in \eqref{def:WigD}. It should be noted that, we use chain rule for the derivation with respect to $\theta$, i.e., $\frac{d}{d\theta} f(\cos \theta) = -\sin \theta \frac{d}{d \cos \theta} f(\cos \theta) $. Additionally, $k$-derivative of Jacobi polynomial is given by $$\frac{d^k}{d \cos^k \theta} = \frac{\Gamma(\xi + \lambda + \alpha + 1 + k)}{2^k \Gamma(\xi  + \lambda + \alpha + 1)} P_{\alpha - k}
^{\xi + k, \lambda + k} (\cos \theta)$$

\subsection{Hypergeometric Polynomials} \label{App:sec2 hypergeo}
In this section, we review some important properties of Jacobi and associated Legendre polynomials that are used in this article.

Jacobi polynomials have a symmetric relation
\begin{equation} \label{App:eq symm_Jacobi}
P_{\alpha}^{\xi,\lambda} (-\cos \theta) = (-1)^{\alpha} P_{\alpha}^{\lambda,\xi} (\cos \theta).
\end{equation}
Furthermore, for $\cos \theta = 1$ and $\cos \theta = -1$, we have
\[ P_{\alpha}^{\xi,\lambda}(1) = \binom{\alpha + \xi}{\alpha}\quad \text{and} \quad P_{\alpha}^{\xi,\lambda}(-1) = (-1)^\alpha\binom{\alpha + \lambda}{\alpha}. \]
 
Similar to Jacobi polynomials, associated Legendre polynomials have symmetric properties  
\begin{equation}\label{App:eq symm_assLeg}
\begin{aligned}
&P_{l}^k(-x)=(-1)^{k+l}P_l^k(x)\\
&P_{l}^{-k}(x)=(-1)^{k}\frac{(l-k)!}{(l+k)!}P_l^k(x).
\end{aligned}
\end{equation}
For degree $l = 0$, Legendre polynomials have the property $P_{0}(x) = 1$. Therefore, from the orthogonal property of Legendre polynomials, we also have
\begin{equation}\label{App:eq int_Leg}
\int_{-1}^{1} P_{*{l}}(x) \mrm{d}x= \int_{-1}^{1} P_{0}(x) P_{l}(x) \mrm{d}x=2\delta_{l0}.
\end{equation}
For $x=1$, the Legendre polynomials $P_l(1)= 1$. Moreover, for $x=-1$, the property can be generated by a symmetric relation of the Legendre polynomials as in \eqref{App:eq symm_assLeg} by setting $k=0$. However, for associated Legendre polynomials, we have $P_l^k (1) = P_l^k*(-1) = 0$ from the following relation
\[
 P_{l}^{k}(x)=(-1)^{*k}(1-x^2)^{k/2}\frac{\mrm{d}^k}{\mrm{d} x^k}P_l(x).
\]
In this article, besides using the definition of the Rodrigues formula, the Legendre polynomials are explicitly given as 
\begin{equation} 
 P_l(x) = 2^l \sum_{h=0}^{l} \binom{l}{h} \binom{\frac{l+h-1}{2}}{l}x^h.
\label{App:eq Legendre_exp}
 \end{equation}
In this work, we use this representation to derive the closed-form sum of equispaced samples Legendre polynomials.  From \eqref{App:eq Legendre_exp}, we can also derive the conditions for $P_l(0)$ for even degree $l$. For $x = 0$, the expansion is non-zero only when $h = 0$
\begin{equation} \label{App:eq even_deg_zeroLeg}
\begin{aligned}
 P_l(0) &= 2^l \binom{l}{0} \binom{\frac{l-1}{2}}{l} = 2^l \frac{\Gamma\left( \frac{1}{2} + \frac{l}{2} \right)}{\Gamma\left(\frac{1}{2} - \frac{l}{2} \right)\enskip l!}= \frac{\left( -1\right)^{l/2}}{2^l}\binom{l}{\frac{l}{2}}, 
\end{aligned}
 \end{equation}
where we have gamma function $\Gamma(l+1) = l!$ and the relation
\begin{equation*}
\begin{aligned}
\Gamma\left(\frac{1}{2} + l \right) &= \frac{\left( 2l\right)! \sqrt{\pi}}{4^l l!}\quad \text{and} \quad
\Gamma\left( \frac{1}{2} - l \right) &= \frac{\left(-4 \right)^l l! \sqrt{\pi}}{\left( 2l\right)!}.
\end{aligned}
\end{equation*}
Another property that is used in this article is the $n-$th derivation of Legendre polynomials. This property can be obtained by using the Gauss hypergeometric function ${}_2F_1(a,b,c,d)$ in \cite[p.101]{bailey_generalized_1964} and in \cite[eq.22-23]{laurent_scaling_2017}. For $x=1$, the relation can be written as
\begin{equation}\label{App:eq deriv_Leg}
\left.\frac{\d^n P_{l}(x)}{\d x^n} \right\vert_{x=1} = \frac{\Gamma(1 +n + l)}{2^n \Gamma\left( n+ 1\right) \Gamma(1 - n + l)}. 
\end{equation}
 
Apart from the derivative of Gauss hypergeometric function, one can also derive from the explicit representation of Legendre polynomials \eqref{App:eq Legendre_exp},
\begin{equation*} 
\begin{aligned} 
\frac{\d^{n-1}}{\d x^{n-1}} P_l(x) &=2^l \sum_{h=n-1}^{l} \binom{l}{h} \binom{\frac{l+h-1}{2}}{l} \frac{h!}{(h-n+1)!} x^{h-n+1}.\\
\end{aligned} 
\end{equation*}

For $x=1$, we obtain the following formula
\begin{equation*} \label{ch:App eq2}
\small
\begin{aligned} 
\left.\frac{\d^{n-1}P_l(x)}{\d x^{n-1}} \right\vert_{x=1} 
&=2^l \binom{l}{n-1} \binom{\frac{l+n-2}{2}}{l} (n-1)! \\ &+ 2^l \sum_{h=n }^{l} \binom{l}{h} \binom{\frac{l+h-1}{2}}{l} \frac{h!}{(h-n+1)!}.
\end{aligned} 
\end{equation*}
For $x=1$ and even $l$,$n$, we have the first term, i.e., $h=n-1$ as
\begin{equation*} 
\begin{aligned} 
2^l \binom{l}{n-1} \binom{\frac{l+n-2}{2}}{l} (n-1)! &= \frac{2^l\Gamma \left(  \frac{ l + n}{2}\right)}{\Gamma \left( l-n+2 \right) \Gamma \left(\frac{  n - l }{2}\right)}.
\end{aligned} 
\end{equation*}
Since $ n \leq l$, this holds due to the property of factorial as in \cite[eq. 6.1.7]{abramowitz_handbook_1965}
\begin{equation*}
\lim_{z \rightarrow n} \frac{1}{\Gamma \left( -z \right)} = \frac{1}{\left( -n -1 \right)!} = 0 \qquad  (n=0,1,2,3,\hdots).
\end{equation*}
Thereby, for even $l$ and $n$, we have
\begin{equation} \label{App:eq derivative_Leg_exp}
\begin{aligned} 
 \left.\frac{\d^{n-1} P_l(x)}{\d x^{n-1}}\right\vert_{x=1}  &= 2^l \sum_{h=n }^{l} \binom{l}{h} \binom{\frac{l+h-1}{2}}{l} \frac{h!}{(h-n+1)!}.
\end{aligned} 
\end{equation}

\subsection{Summations}
Several summations have been used in this work, mainly to prove main result, for example the Bernoulli or Faulhaber summation and Abel partial summation. In this section, we will provide a concise summary of these summations.
\subsubsection{Bernoulli Summation}
Suppose we have the summation $\sum_{p=1}^m p^k$ for integer $k$. The expression of this summation is originally introduced by Faulhaber until $k = 17$ and later generalized by Bernoulli \cite{knuth_johann_1993, jacobi_c_2013}
\begin{equation}\label{App:eq Faul_sum}
\sum_{p=1}^m p^k = \frac{1}{k+1}\sum_{j=0}^k{B_j}\binom{k+1}{j} {m}^{k+1-j},
\end{equation}
where $B_j$ is the Bernoulli number. For clarity, we list some Bernoulli numbers as given in the following table \footnote{In this work, we are using a convention where $B_2 = 1/2$ and the summation \eqref{App:eq Faul_sum} also follows this convention.}
\begin{center} 
    \begin{tabular}{| l | l | l | l | l | l | l | l | l |}
    \hline
    Index ($j$) & 0 &1 &2 &3& 4& 5& 6 &7 \\ \hline
   $B_j$ & 1 & 1/2 & 1/6 & 0  & -1/30 & 0 & 1/42 & 0  \\
   \hline
    \end{tabular}.\label{App:tab bernoulli_table}
\end{center}
It can be seen that for odd $j \geq 3$, the Bernoulli number is equal to zero. This property is useful to prove Lemma \ref{ch5:lemm sum_equi}, Proposition \ref{ch5:prop residual}, and Lemma \ref{ch5:lemm increasing}. Another fascinating property of this number is its close relation with the Riemannian zeta function. We show that the relation for even $j \geq 2$ is 
\begin{equation}\label{App:eq bernoulli_zeta}
B_j = \frac{(-1)^{\frac{j}{2}+1} 2 j!}{(2\pi)^j} \zeta(j).
\end{equation}
 The zeta function also has an interesting property for some specific values, as given in the following table
\begin{center}
    \begin{tabular}{| l | l | l | l | l |}  
    \hline
    Index ($j$) & 2 & 4 & 6 & 8 \\ \hline
   $ \zeta(j)$ & $\pi^2/6$ & $\pi^4/90$ & $\pi^6/945$ & $\pi^8/9450$ \\
   \hline
    \end{tabular}.\label{App:tab zeta_table}
     
\end{center}
It is obvious that for sufficiently large $j$, the zeta function converges to $1$.

\subsubsection{Abel Partial Summation}
The Abel partial summation is defined by Niels Henrik Abel \cite{abel_untersuchungen_1895}, which has a similar property to the integration by parts. Suppose we have $n \in \mbb N$ with sequences $a_1,a_2,\hdots,a_n$ and also $b_1,b_2,\hdots,b_n \in \mbb R$ with $A_p = a_1 + a_2 + \hdots + a_p$, then we have
\begin{equation}\label{App: Abelsum}
\sum_{p=1}^n a_p b_p = A_n b_n + \sum_{p=1}^{n-1} A_p \left( b_p - b_{p+1}\right).
\end{equation}
 
\subsection{Properties of Wigner 3j Symbols}\label{App:Property_Wigner3j}
An explicit formula for the general Wigner 3j symbols can be seen in most angular momentum literature. In this paper, the explicit formula for Wigner 3j symbols is taken from \cite{kennedy_hilbert_2013} 
\begin{equation}
\begin{aligned} \label{App:eq general_wigner3j}
\begin{pmatrix}
   l_1 & l_2 & l_3 \\
   k_1 & k_2 & k_3
  \end{pmatrix}  &=  s^{l_1,l_2}_{l_3,k_1,k_2}\frac{(-1)^{l_1 - l_2 -k_3}}{\sqrt{2 l_3 + 1}} s^{l_1,l_2}_{l_3,k_1,k_2},
\end{aligned}
\end{equation}
where the value $s^{l_1,l_2}_{l_3,k_1,k_2}$ is given by
\begin{equation}
\small
\begin{aligned} 
& \sqrt{\frac{ (l_3 + l_1 - l_2)!(l_3 - l_1 + l_2)!(l_1 + l_2 - l_3)!(l_3 - k_3)!(l_3 + k_3)!}{ (l_1 + l_2 +l_3 + 1)!(l_1 - k_1)!(l_1 + k_1)!(l_2 - k_2)!(l_2 + k_2)!}} \\
&\times \sum_{t} \frac{(-1)^{t + l_2 + k_2}\sqrt{2l_3+1}(l_3 +l_2 +k_1 -t)!(l_1 -k_1 + t)!}{(l_3 - l_1 + l_2 -t)!(l_3 -k_3 - t)! (t)!(l_1 - l_2 + k_3 + t)!}.
\end{aligned}
\end{equation}
The sum over $t$ is chosen such that all variables inside the factorial are non-negative. There are several conditions that make the expression simpler, for example the condition $k_1 = k_2 =  k_3 = 0$ which is frequently used in this paper. In this case, Wigner 3j symbols are explicitly given by
\begin{equation}
\begin{aligned} \label{App:eq special_wigner3j}
\begin{pmatrix}
   l_1 & l_2 & l_3 \\
   0 & 0 & 0
  \end{pmatrix}  &=  (-1)^{L} \Biggl( \frac{(2L-2l_1)!(2L-2l_2)!(2L-2l_3)!}{(2L+1)!}\Biggr)^{\frac{1}{2}}\\&\Biggl( \frac{L!}{(L-l_1)!(L-l_2)!(L-l_3)!}\Biggr),
\end{aligned}
\end{equation}
where $2L=l_1+l_2+l_3$ is an even integer. 


\end{document}